%% file: one-column.tex
\documentclass[lettersize,onecolumn]{IEEEtran}
\usepackage{algorithmic}
\usepackage{algorithm}
\usepackage[caption=false,font=normalsize,labelfont=sf,textfont=sf]{subfig}
\usepackage{textcomp}
\usepackage{stfloats}
\usepackage{url}
\usepackage{verbatim}
\usepackage{graphicx}
\usepackage{cite}

\input{config}

\begin{document}

\title{Convergence of Denoising Diffusion Probabilistic Models: A \Renyi Divergence Perspective}

\author{Wenliang Liu, Santiago Mazuelas,~\IEEEmembership{Senior Member,~IEEE}, Guanding Yu,~\IEEEmembership{Senior Member,~IEEE}, \\ Renjie Liao, Lele Wang,~\IEEEmembership{Member,~IEEE}, and Yuan Shen\textsuperscript{\textsection},~\IEEEmembership{Senior Member,~IEEE}}



\maketitle

\begingroup
\renewcommand\thefootnote{}
\footnotetext{Wenliang Liu is with the Department of Electronic Engineering, Tsinghua University, Beijing 100084, China (e-mail: liuwenliang@mail.tsinghua.edu.cn). This work was done in part while he was with the College of Information Science and Electronic Engineering, Zhejiang University, Hangzhou 310027, China (e-mail: liuwenliang@zju.edu.cn).}
\footnotetext{Santiago Mazuelas is with the Basque Center for Applied Mathematics (BCAM) and with IKERBASQUE Foundation, Bilbao 48009, Spain (e-mail: smazuelas@bcamath.org).}
\footnotetext{Guanding Yu is with the College of Information Science and Electronic Engineering, Zhejiang University, Hangzhou 310027, China (e-mail: yuguanding@zju.edu.cn).}
\footnotetext{Renjie Liao is with the Department of Electrical and Computer Engineering, The University of British Columbia, Vancouver, BC V6T 1Z4, Canada, and also with the Vector Institute, Toronto, ON M5G 0C6, Canada (e-mail: rjliao@ece.ubc.ca).}
\footnotetext{Lele Wang is with the Department of Electrical and Computer Engineering, The University of British Columbia, Vancouver, BC V6T 1Z4, Canada (e-mail: lelewang@ece.ubc.ca).}
\footnotetext{Yuan Shen is with the Department of Electronic Engineering, Tsinghua University, Beijing 100084, China (e-mail: shenyuan\_ee@tsinghua.edu.cn).}
\renewcommand\thefootnote{\textsection}
\footnotetext{Corresponding author.}
\endgroup
\begin{abstract}
	\Acp{ddpm} have emerged as a powerful framework for generative modeling by transforming standard Gaussian inputs into outputs that approximate a target distribution. 
	While \acp{ddpm} have been proved to converge in total variation, KL divergence, and Wasserstein distance, these discrepancy measures share a fundamental limitation of being insensitive to distributional mismatch in low-probability regions.
	To address this, we adopt the stronger \Renyi divergence as the discrepancy measure. Leveraging a Girsanov-type change-of-measure argument, we establish the first \Renyi characterization of \acp{ddpm} to the best of our knowledge, showing that the total error decomposes into a mixing error, a discretization error, and a score estimation error. 
	\lwl{We assume the target distribution is sub-Gaussian and characterize the score error via the cumulant generating function, relaxing both the log-Sobolev-inequality and $L^\infty$-score-error conditions used in prior \Renyi divergence analyses of inexact Langevin Monte Carlo, a score-based sampling algorithm closely related to \ac{ddpm}.}
	Our results also establish an iteration complexity for achieving $\epsilon$-accuracy in $\alpha$-\Renyi divergence that is asymptotically proportional to $\alpha^3/\epsilon$.
	Finally, by extending our analysis, we establish the first \Renyi divergence convergence guarantee for critically damped \acp{ddpm}, a kinetic variant governed by critically damped Langevin dynamics. Similar to existing characterizations of critically damped \acp{ddpm}, this result shares the structural scaling of standard \acp{ddpm}, differing by at most a logarithmic factor in the mixing error exponent, and supports the conjecture that critically damped \acp{ddpm} do not yield a theoretical improvement.
\end{abstract}
\acresetall

\begin{IEEEkeywords}
\Renyi divergence, denoising diffusion probabilistic models, score-based generative models, critically damped Langevin dynamics.
\end{IEEEkeywords}

\input{main_body}

%

\appendices
\input{appendix}

\bibliographystyle{IEEEtran}
\bibliography{references.bib}

\end{document}

%% file: config.tex
\usepackage[utf8]{inputenc} 
\usepackage[T1]{fontenc}    
\usepackage{hyperref}       
\usepackage{booktabs}       
\usepackage{amsmath,amsfonts,amssymb}       
\usepackage{nicefrac}       
\usepackage{microtype}      
\usepackage{graphicx}
\usepackage{subcaption}
\usepackage{caption}
\usepackage{amsthm}
\usepackage{mathtools}
\usepackage{bbm}
\usepackage{xspace}
\usepackage{mathrsfs}
\usepackage{booktabs}
\usepackage{array}
\usepackage{xcolor}
\usepackage{scalerel}
\usepackage{etoolbox}
\usepackage{multirow}
\usepackage{makecell}
\usepackage[inline]{enumitem}
\usepackage{acronym}
\usepackage{algorithmic}
\usepackage{algorithm}
\usepackage[caption=false,font=normalsize,labelfont=sf,textfont=sf]{subfig}
\usepackage{textcomp}
\usepackage{stfloats}
\usepackage{url}
\usepackage{verbatim}
\usepackage{graphicx}

\newtheorem{theorem}{Theorem}
\newtheorem{lemma}[theorem]{Lemma}
\newtheorem{proposition}[theorem]{Proposition}
\newtheorem{corollary}[theorem]{Corollary}
\newtheorem{definition}[theorem]{Definition}
\newtheorem{assumption}[theorem]{Assumption}

\newtheorem{remark}{Remark}

\renewcommand{\d}{\mathrm{d}}
\newcommand{\Renyi}{R\'enyi }
\newcommand{\Holder}{H\"older }

\newcommand{\OU}{{\operatorname{OU}}}
\newcommand{\CD}{{\operatorname{CD}}}%

\newcommand{\Aux}{{\operatorname{Aux}}}

\newcommand{\RenyiDivergence}[3][\alpha]{\mathcal{R}_{#1}\left(#2||#3\right)}
\newcommand{\KLDivergence}[2]{\mathcal{D}_{\mathrm{KL}}\left(#1||#2\right)}

\newcommand{\WassersteinDistance}[3][2]{\mathcal{W}_{#1}\left(#2,#3\right)}
\newcommand{\squaredWassersteinDistance}[3][2]{\mathcal{W}^2_{#1}\left(#2,#3\right)}
\newcommand{\totalVariation}[2]{\operatorname{TV}\left(#1,#2\right)}
\newcommand{\squaredTotalVariation}[2]{\operatorname{TV}^2\left(#1,#2\right)}
\newcommand{\ent}[2]{\mathrm{Ent}_{#1}\left(#2\right)}
\newcommand{\expectation}[2]{\bbE_{#1}\left[#2\right]}
\newcommand{\exponential}[1]{\exp\left(#1\right)}
\newcommand{\scoreError}{\varepsilon^2_{\operatorname{score}}}
\newcommand{\scoreErrorCD}{\varepsilon^2_{\operatorname{score,CD}}}
\newcommand{\innerProd}[2]{\left\langle #1,#2\right\rangle}
\newcommand{\RadonNikodym}[2]{\frac{\d #1}{\d #2}}
\newcommand{\bbR}{\mathbb{R}}
\newcommand{\bbE}{\mathbb{E}}

\newcommand{\1}[1]{\mathbbm{1}_{#1}}
\newcommand{\cE}{\mathcal{E}}
\newcommand{\cF}{\mathcal{F}}
\newcommand{\cI}{\mathcal{I}}
\newcommand{\cN}{\mathcal{N}}
\newcommand{\cO}{\mathcal{O}}
\newcommand{\bigTheta}{\varTheta}

\newcommand{\bzero}{\mathbf{0}}

\newcommand{\lwl}[1]{\textcolor{black}{#1}}
\renewcommand{\bar}{\overline}
\renewcommand{\tilde}{\widetilde}

\acrodef{ddpm}[DDPM]{denoising diffusion probabilistic model}
\acrodef{sde}[SDE]{stochastic differential equation}
\acrodef{cgf}[CGF]{cumulant generating function}
\acrodef{lsi}[LSI]{log-Sobolev inequality}
\acrodef{ou}[OU]{Ornstein-Uhlenbeck}
\acrodef{kl}[KL]{Kullback-Leibler}
\acrodef{tv}[TV]{total variation}
\acrodef{kls}[KLS]{Kannan–Lovász–Simonovits}
\acrodef{lmc}[LMC]{Langevin Monte Carlo}
\makeatletter
\DeclareRobustCommand\onedot{\futurelet\@let@token\@onedot}
\def\@onedot{\ifx\@let@token.\else.\null\fi\xspace}
 
\def\eg{e.g\onedot} 
\def\ie{i.e\onedot} 
\def\cf{cf\onedot} 
\def\st{s.t\onedot}

\makeatother

%% file: main_body.tex
\section{Introduction}
\label{section:Introduction}

\Acp{ddpm}~\cite{ho2020denoisingdiffusionprobabilisticmodels, song2021scorebasedgenerativemodelingstochastic, nichol2021improveddenoisingdiffusionprobabilistic}, as a class of score-based generative models, have emerged as a prominent paradigm for sampling from high-dimensional and complex probability distributions. 
The framework operates by gradually corrupting data with noise until it converges to a standard Gaussian distribution. 
During this process, the score function (a high-dimensional, time-varying vector field) is approximated from the data. The learned score function is then utilized to simulate the reverse diffusion process, transforming Gaussian inputs into samples that closely approximate the target distribution. Notably, \acp{ddpm} have achieved remarkable success in synthesizing high-quality images~\cite{hoogeboom2023simplediffusionendtoenddiffusion}, audio~\cite{kong2021diffwaveversatilediffusionmodel}, and video~\cite{NEURIPS2022_39235c56}.

The remarkable empirical success of \acp{ddpm} has driven extensive investigation into their theoretical guarantees. Existing studies establish that, given accurate $L^2$-approximations of the score function, \acp{ddpm} converge to the target distribution in \ac{tv} distance~\cite{chen2022sampling}, \ac{kl} divergence~\cite{chen2023improved,conforti2024klconvergenceguaranteesscore}, or Wasserstein distance~\cite{silveri2025beyond,strasman2026wasserstein}. 
Nevertheless, as these standard discrepancy measures penalize distributional discrepancies moderately, a small overall error might conceal severe, localized mismatches between the generated and target distributions. Crucially, this issue becomes far more severe in low-probability regions, because these discrepancy measures are mathematically formulated to weigh distributional errors according to absolute probability mass. To provide stronger and more reliable theoretical guarantees, this paper characterizes the convergence behavior of \acp{ddpm} under the \Renyi divergence (\cf Definition~\ref{definition:Renyi divergence}).

Adopting the \Renyi divergence to evaluate distributional discrepancies offers three advantages over the aforementioned measures. First, it is significantly more sensitive because it imposes a more stringent penalty on distributional mismatches. Consequently, achieving convergence under the \Renyi divergence yields a stronger guarantee that the learned distribution accurately aligns with the target. This enhanced sensitivity is reflected in its ability to dominate other standard discrepancy measures. 
Specifically, as an increasing function of the order $\alpha$, the \Renyi divergence for $\alpha>1$ strictly upper bounds the \ac{kl} divergence, which coincides with the \Renyi divergence at $\alpha=1$. Consequently, it also dominates the squared \ac{tv} distance by Pinsker's inequality ($\KLDivergence{P}{Q}\geq 2\squaredTotalVariation{P}{Q}$), as well as the squared $1$- or $2$-Wasserstein distances under conditions such as Talagrand's transport inequalities~\cite{villani2009optimal}. Such conditions are necessary to compare the two discrepancy measures, since the Wasserstein distance depends on the underlying geometry, while the \Renyi divergence does not. 
Further examples demonstrating that the \Renyi divergence can be more sensitive than the Wasserstein distance are provided in Fig.~\ref{fig:Renyi divergence is more sensitive} and Proposition~\ref{proposition:KL-close distributions can diverge under Renyi divergence}.

\begin{figure}[tb]
	\centering
	\begin{minipage}[c]{\textwidth}
		\centering
		\small
		\captionof{table}{Scales of discrepancy measures for the distributions in Fig.~\ref{fig:Renyi divergence is more sensitive}.}
		\begin{tabular}{@{}lccccc@{}}
			\toprule
			\textbf{Discrepancy Measure} 
			& $\squaredTotalVariation{P}{Q}$
			& $\squaredWassersteinDistance[1]{P}{Q}$
			& $\squaredWassersteinDistance[2]{P}{Q}$
			& $\KLDivergence{P}{Q}$
			& $\RenyiDivergence{P}{Q}$ \\
			\midrule
			\textbf{Scale} 
			& $\cO(\epsilon^2)$ 
			& $\epsilon^2 m^2$ 
			& $\cO(\epsilon m^2)$ 
			& $\cO(\epsilon m^2)$ 
			& $\dfrac{\alpha}{2} m^2 - \dfrac{\alpha}{\alpha - 1} \log\left(\dfrac{1}{\epsilon}\right) + \cO\left(e^{-(\alpha-1)m^2}\right)$ \\
			\bottomrule
		\end{tabular}
		\label{tab:metrics_scaling}
		\vspace{1.0em}
	\end{minipage}
	
	\begin{minipage}[c]{0.49\textwidth}
		\centering
		\includegraphics[width=\linewidth]{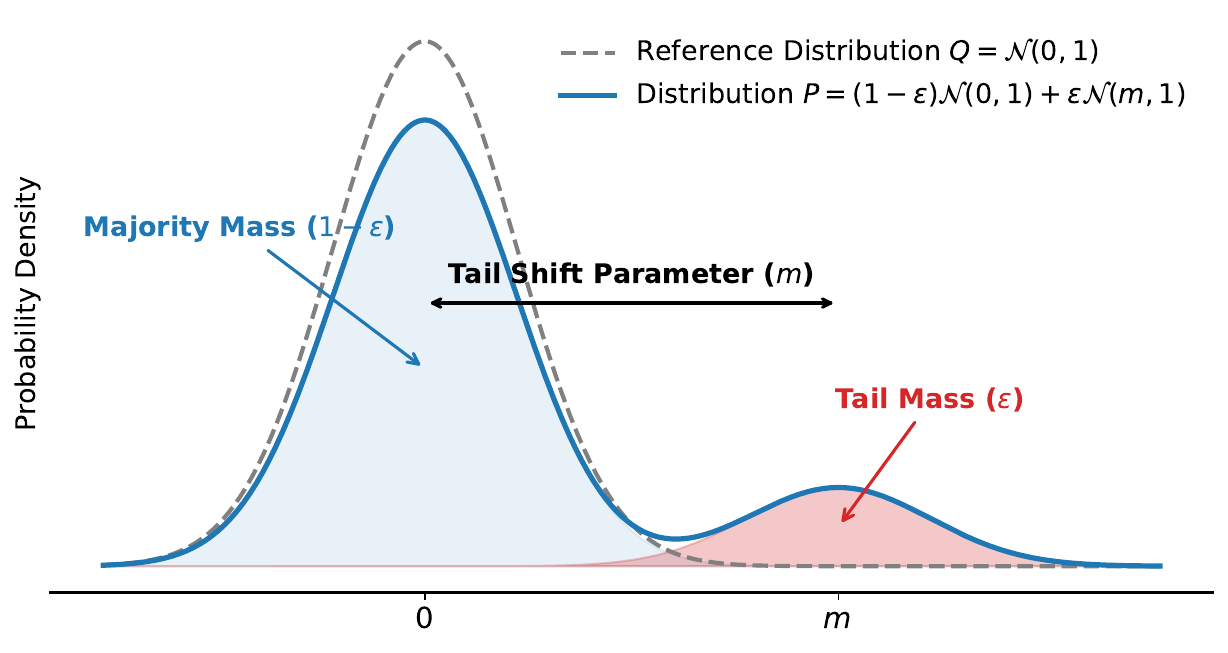}
	\end{minipage}%
	\hfill
	\begin{minipage}[c]{0.49\textwidth}
		\centering
		\includegraphics[width=\linewidth]{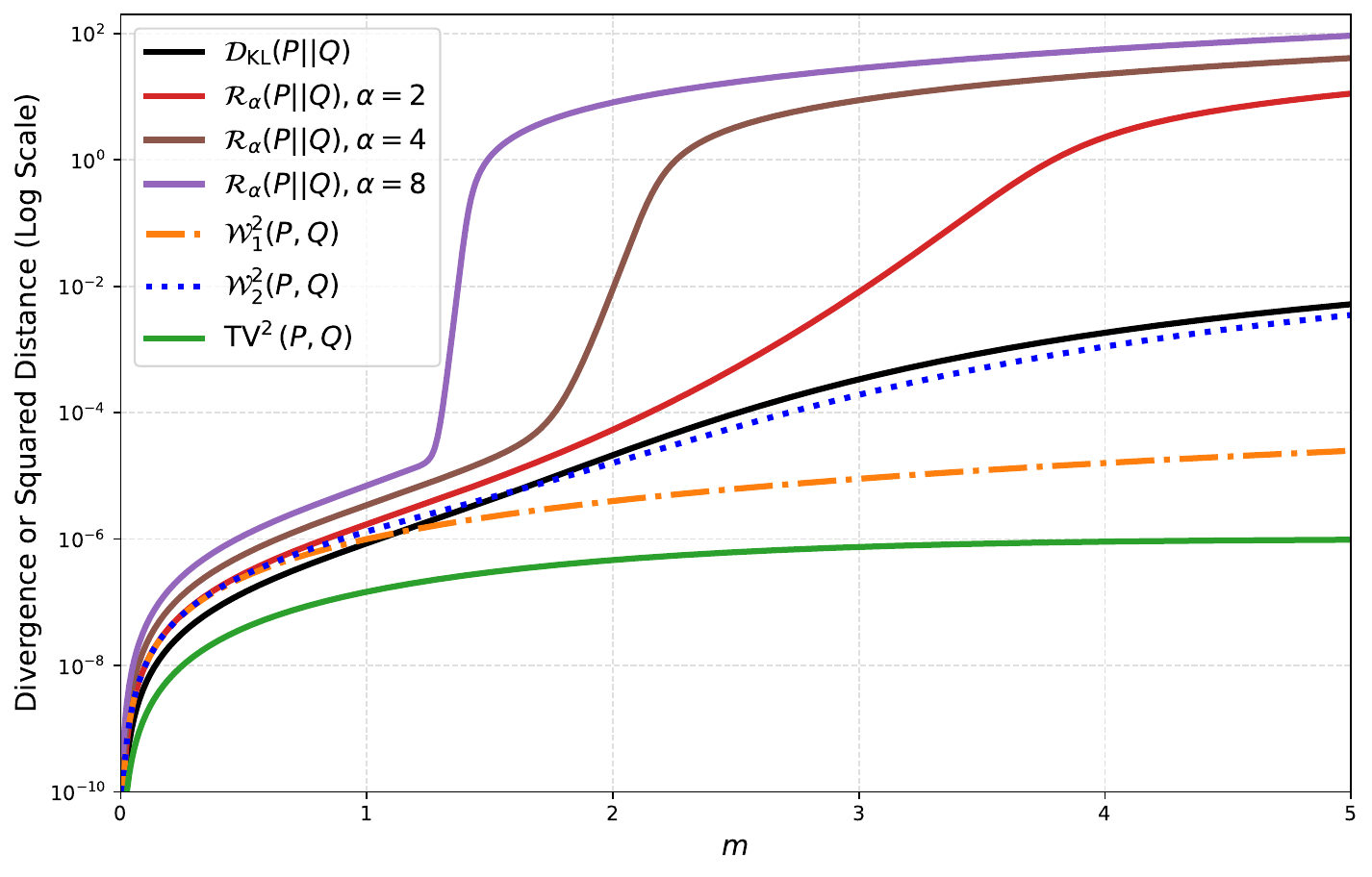}
	\end{minipage}
	\caption{\Renyi divergence is more sensitive to distributional discrepancies, particularly in low-probability regions (large $m$). In this example, we compare a Gaussian mixture distribution $P = (1-\epsilon)\cN(0,1) + \epsilon \cN(m,1)$ against a standard Gaussian reference distribution $Q = \cN(0,1)$. Fixing $\epsilon = 0.001$, we evaluate the $\alpha$-\Renyi divergence against other standard discrepancy measures as the tail shift parameter $m$ varies. Although most of these measures scale quadratically with $m$ (Table~\ref{tab:metrics_scaling}), the \Renyi divergence admits a larger coefficient ($\alpha / 2$), allowing it to dominate the \ac{kl} divergence and squared \ac{tv} for $\alpha>1$, and surpass the squared $1$- and $2$-Wasserstein distances for $\alpha>2$. Most importantly, the \Renyi divergence exhibits an additive logarithmic dependence on the tail mass $\epsilon$, contrasting with the multiplicative polynomial scaling of the other discrepancy measures (Table~\ref{tab:metrics_scaling}). 
	This fundamental distinction makes the Rényi divergence less susceptible to suppression by small values of $\epsilon$, enabling it to achieve significantly larger values for large $m$. Further technical details are deferred to Appendix~\ref{section:appendix:details_of_the_figure}.}
	\label{fig:Renyi divergence is more sensitive}
%
\end{figure}

Second, the \Renyi divergence is particularly effective at highlighting discrepancies in low-probability regions. This characteristic makes it advantageous for safety-critical applications, such as autonomous driving and automated medical diagnosis, where accurately modeling rare events is paramount. Mathematically, two distributions that appear nearly identical under \ac{tv}, \ac{kl} divergence, or Wasserstein distance can still be drastically far apart under the \Renyi divergence (\cf Proposition~\ref{proposition:KL-close distributions can diverge under Renyi divergence}). This phenomenon occurs because weaker discrepancy measures impose only a moderate penalty on distributional mismatch, tending to smooth over large relative errors in regions where absolute probability mass is small. In contrast, the \Renyi divergence penalizes distributional mismatch more severely, ensuring that critical deviations in tail regions are not neglected. To demonstrate this behavior, we consider a motivating example illustrating how each discrepancy measure responds to a distributional mismatch caused by a probability mass $\epsilon$ in a low-probability region. As shown in Table~\ref{tab:metrics_scaling}, the $\alpha$-\Renyi divergence depends on the tail mass $\epsilon$ through an additive logarithmic term $\log(1/\epsilon)$, whereas all other discrepancy measures scale multiplicatively (either linearly or quadratically) with $\epsilon$. Consequently, the \Renyi divergence is less suppressed by small values of $\epsilon$, allowing it to impose a much heavier penalty on distributional discrepancies in tail regions than other measures. An illustrative comparison of these quantities is provided in Fig.~\ref{fig:Renyi divergence is more sensitive}.

Third, the \Renyi divergence provides inherent flexibility because its order $\alpha$ can be precisely tuned to the demands of a specific task. This parameter controls the ``granularity'' of the convergence behavior, allowing one to explicitly balance distributional similarity between macroscopic (global) and microscopic (local) scales. At the lower extreme ($\alpha = 1$), the \Renyi divergence reduces to the \ac{kl} divergence, which encourages the learned distribution to globally resemble the target. At the upper extreme ($\alpha = \infty$), it becomes the essential supremum of the logarithmic density ratio. Convergence in this regime provides a strict worst-case guarantee, ensuring that the learned distribution aligns with the target almost everywhere. For intermediate values ($1 < \alpha < \infty$), the order $\alpha$ acts as an interpolation parameter, smoothly trading off between relaxed global structural matching and stringent local pointwise matching.

\lwl{Despite these benefits, the \Renyi divergence error of \acp{ddpm} remains unexplored. Existing \Renyi guarantees concern \ac{lmc}~\cite{vempala2019rapid, chewi2025analysis, NEURIPS2020_50cf0fe6, pmlr-v151-erdogdu22a}, a class of Langevin-dynamics-based sampling algorithms that assume an exact, time-independent score, as opposed to the estimated, time-varying score in \acp{ddpm} (\cf~\hyperlink{para:renyi-related-work}{Section~\ref{section:Related Work}, ``R\'enyi divergence''}). Moreover, these \Renyi analyses for \ac{lmc} rely on unrealistic or hard-to-verify conditions, such as requiring the target distribution to satisfy a \ac{lsi}, or assuming an $L^\infty$-bounded score error whenever the exact score is unavailable.}

\lwl{In this paper, we establish the first Rényi-divergence characterization of \acp{ddpm}. Our specific contributions are summarized as follows.}
\begin{itemize}
	\item Leveraging a Girsanov-type change-of-measure argument, we derive an upper bound on the \Renyi divergence between two diffusion processes with mismatched drift coefficients. This result provides an analytical framework applicable beyond \acp{ddpm} to broader score-based generative models.
	
	\item Applying this framework to \acp{ddpm}, we characterize their \Renyi divergence error under a sub-Gaussian target distribution and a Lipschitz continuous score function. The error bound can be decomposed into a mixing error, a discretization error, and a score error evaluated via the \ac{cgf}. To the best of our knowledge, this constitutes the first characterization of the \Renyi divergence for \acp{ddpm}.
	
	\item \lwl{We characterize the iteration complexity for achieving $\epsilon$-accuracy in $\alpha$-\Renyi divergence. For $\alpha \gg 1$ and $\epsilon \ll 1$, this complexity is asymptotically proportional to $\alpha^3/\epsilon$. Hence, controlling the stronger \Renyi divergence costs at most an additional multiplicative factor of $\alpha^3$. In addition, the linear $1/\epsilon$ dependence is consistent with established guarantees in (squared) \ac{tv}, \ac{kl} and (squared) Wasserstein distance.}
	
	\item We extend our framework to critically damped \acp{ddpm}, a kinetic version of standard \acp{ddpm} governed by critically damped Langevin dynamics. To the best of our knowledge, we establish the first \Renyi divergence error bound for critically damped \acp{ddpm}. Up to a logarithmic factor in the effective running time, the resulting bound scales identically to its standard \ac{ddpm} counterpart, supporting the conjecture that critically damped \acp{ddpm} do not yield a theoretical improvement over standard \acp{ddpm}.
\end{itemize}

The remainder of this paper is organized as follows. 
Section~\ref{section:Related Work} reviews related work. 
Section~\ref{section:Models} introduces the standard \ac{ddpm} formulation, formally defines the \Renyi divergence, and provides theoretical evidence of its enhanced sensitivity, further motivating our analysis. Section~\ref{section:Main Results} presents our main results, establishing convergence guarantees for standard \acp{ddpm} and extending the analysis to kinetic variants governed by critically damped Langevin dynamics. Section~\ref{section:Discussion} contextualizes our findings within the existing literature and outlines promising directions for future work. 
Finally, Section~\ref{section:Conclusions} concludes the paper.

\subsection*{Notations}
For probability distributions $P$ and $Q$, $\RadonNikodym{P}{Q}$ denotes their Radon-Nikodym derivative. 
We denote the $d$-dimensional standard Gaussian distribution by $\xi^d$ and a $d$-dimensional standard Gaussian vector by $\Xi^d$ (\ie, $\Xi^d\sim\xi^d$).
For a stochastic process $(X_t)_{t\in[0,T]}$ with path measure $\rho$, we denote its marginal distribution at time $t$ by $\rho_t$, so that $X_t\sim \rho_t$.
For real numbers $a$ and $b$, $a\wedge b$ and $a\vee b$ denote $\min\left\{a,b\right\}$ and $\max\left\{a,b\right\}$, respectively. For functions $f$ and $g$, the notation $f \lesssim g$ indicates that $f\leq Cg$ for a universal constant $C>0$, and $f\asymp g$ indicates that $f \lesssim g$ and $g\lesssim f$ hold simultaneously. We use the standard notations $\cO(\cdot)$ and $\bigTheta(\cdot)$ to represent the asymptotic behaviors corresponding to $\lesssim$ and $\asymp$, respectively. The notation $\widetilde{\cO}(\cdot)$ and $\widetilde{\bigTheta}(\cdot)$ denote their respective counterparts that suppress polylogarithmic factors (\ie, $\widetilde{\cO}(x) = \cO(x\log^n(x))$ for some $n\in \mathbb{N}$).
Finally, for any time-dependent function $f(t)$ defined on $[0,T]$, we denote its time reversal by $\bar{f}(t) \coloneqq f(T-t)$.

\section{Related Work}
\label{section:Related Work}

The conceptual foundations of \acp{ddpm} can be traced back to \cite{sohldickstein2015deepunsupervisedlearningusing, ho2020denoisingdiffusionprobabilisticmodels}. Later, the work in \cite{song2021scorebasedgenerativemodelingstochastic} introduced a continuous-time perspective, recasting \acp{ddpm} as score-based generative models~\cite{song2020generativemodelingestimatinggradients} driven by a \ac{sde}.
Existing analysis mostly focuses on establishing the convergence of \acp{ddpm} to their target distribution in terms of \ac{tv} distance, \ac{kl} divergence, or Wasserstein distance.
\paragraph{\ac{tv} Distance}
	The first theoretical guarantees for score-based generative models were established in \ac{tv} distance by~\cite{debortoli2023diffusionschrodingerbridgeapplications}. However, these results rely on $L^\infty$ score error assumptions and stringent dissipative conditions, and suffer from exponential growth with respect to the running time. Subsequently, a breakthrough was made in~\cite{lee2022convergence}, establishing the first polynomial-time \ac{tv} convergence bounds under more realistic $L^2$ score errors, albeit only for target distributions satisfying an \ac{lsi}. The \ac{lsi} requirement was later relaxed in~\cite{lee2023convergence}, provided that the target distribution exhibits sufficient tail decay and smoothness. In the seminal work of~\cite{chen2022sampling}, \ac{tv} bounds were derived for general data distributions requiring only a finite second moment, achieving an iteration complexity that matches that of Langevin Monte Carlo sampling algorithms. Furthermore, the authors utilized an early stopping technique to prevent the \ac{tv} bound from becoming vacuous under the manifold hypothesis~\cite{DBLP:journals/tmlr/Bortoli22}.

\paragraph{\ac{kl} Divergence}
	Subsequent research has focused increasingly on \ac{kl} divergence bounds~\cite{chen2023improved, benton2024nearlydlinearconvergencebounds, Strasman2024AnAO, conforti2024klconvergenceguaranteesscore, yang2022convergence, pmlr-v99-tzen19a}, as it provides a stronger discrepancy measure than \ac{tv} distance by Pinsker's inequality. The work in~\cite{chen2023improved} demonstrates that the \ac{tv} bounds established in~\cite{chen2022sampling} can be improved to \ac{kl} divergence without introducing additional assumptions. Notably, \ac{kl} bounds can still be established even when the score smoothness assumption is relaxed to hold only at time $t=0$, or removed entirely by compensating with an early-stopping scheme. The cost of relaxing score smoothness is a quadratic dependence on the dimension $d$ in the iteration complexity required to achieve a desired accuracy. This dimensional dependence was subsequently tightened to nearly linear in $d$ via stochastic localization~\cite{benton2024nearlydlinearconvergencebounds}, a technique originating in~\cite{eldan2013thin} that has driven recent advances on the \ac{kls} conjecture~\cite{chen2021constantlowerboundisoperimetric}. Recently, state-of-the-art analyses have successfully removed the need for both score smoothness assumptions and early stopping entirely. These works establish sharp \ac{kl} convergence for any distribution, requiring only that the target distribution possesses finite Fisher information with respect to the standard Gaussian measure~\cite{Strasman2024AnAO, conforti2024klconvergenceguaranteesscore}.
	
\paragraph{Wasserstein Distance} 
	The analysis of Wasserstein distance is relatively independent of the \ac{tv}- and \ac{kl}-based lines of work, as these discrepancy measures are not comparable in general. Early Wasserstein bounds relied on stringent dissipativity conditions~\cite{block2022generativemodelingdenoisingautoencoders} or suffered from exponential dependence on parameters such as the manifold diameter~\cite{DBLP:journals/tmlr/Bortoli22}. Subsequent work~\cite{lee2023convergence} extended 2-Wasserstein guarantees to general distributions with bounded support. However, extending such results to non-compact domains has largely required the target distribution to be strongly log-concave in order to control the Wasserstein distance~\cite{gao2025wasserstein, bruno2025on, Strasman2024AnAO}. This restriction limits the theory's applicability to realistic scenarios. More recent work has begun to relax the log-concavity assumption. For example, the work in~\cite{silveri2025beyond} studies weakly convex potential functions, while the work in~\cite{strasman2026wasserstein} allows the potential function to deviate from strong convexity via a Lipschitz perturbation.
	
\hypertarget{para:renyi-related-work}{}
\paragraph{\Renyi divergence}
\lwl{In contrast to the extensive body of work in \ac{tv}, \ac{kl}, and Wasserstein distances, the \Renyi divergence error for \acp{ddpm} remains unexplored. Existing \Renyi analyses concern \ac{lmc}~\cite{vempala2019rapid, chewi2025analysis, NEURIPS2020_50cf0fe6, pmlr-v151-erdogdu22a}, which differs from the \ac{ddpm} setting in three respects.
	\begin{enumerate*}
		\item \ac{lmc} assumes exact knowledge of the score function, which needs to be estimated for \acp{ddpm}.
		\item The score of \ac{lmc} is time-independent, as opposed to the time-dependent score of \acp{ddpm}.
		\item \ac{lmc} does not reverse any stochastic process, and hence does not need to account for the mixing error present in \acp{ddpm}.
	\end{enumerate*}
	Consequently, the \Renyi analyses for \ac{lmc} do not transfer to \acp{ddpm}. 
	Moreover, these works typically require the target distribution to be strongly log-concave~\cite{NEURIPS2020_50cf0fe6}, to satisfy a strong dissipative condition~\cite{pmlr-v151-erdogdu22a}, or to satisfy an isoperimetric inequality, such as the Poincar\'e inequality or \ac{lsi}~\cite{vempala2019rapid}, or more generally, a Latała--Oleszkiewicz inequality that interpolates between the Poincar\'e and \ac{lsi}~\cite{chewi2025analysis}. These conditions are either unrealistic or hard to verify in practice. The work of~\cite{yang2022convergence} considered the inexact \ac{lmc} where the exact score is unavailable, bringing it closer to the \ac{ddpm} setting. However, their analysis characterizes the score error via the restrictive $L^\infty$ norm while also assuming the target distribution satisfies an \ac{lsi}.
}
\paragraph{Critically Damped \acp{ddpm}}
	Compared with standard \acp{ddpm}, the theoretical understanding of critically damped \acp{ddpm}~\cite{dockhorn2021score} remains comparatively limited. Existing works have established the convergence of critically damped \acp{ddpm} in \ac{tv} distance~\cite{chen2022sampling}, \ac{kl} divergence~\cite{conforti2024klconvergenceguaranteesscore}, and Wasserstein distance~\cite{strasman2026wasserstein}. Nevertheless, no explicit evidence supports that critically damped \acp{ddpm} achieve a theoretical improvement over their standard counterpart, despite their empirical success.

Finally, as a side remark, theoretical analysis of other generative paradigms, such as flow matching~\cite{lipman2023flowmatchinggenerativemodeling, liu2022flowstraightfastlearning} and diffusion bridges~\cite{debortoli2023diffusionschrodingerbridgeapplications, zhou2023denoisingdiffusionbridgemodels}, has also progressed rapidly. See, e.g.,~\cite{benton2024errorboundsflowmatching, su2025flowmatchingkldivergence, albergo2025stochasticinterpolantsunifyingframework, tang2026foundationsschrodingerbridgesgenerative, silveri2026diffusionflowmatchingdimensionimproved}.

\section{Problem Formulation}
\label{section:Models}

This section first formulates the \ac{ddpm} in Section~\ref{subsection:DDPM}. Section~\ref{subsection:Renyi divergence} then formally defines the \Renyi divergence and provides a quantitative characterization of its enhanced sensitivity.

\subsection{Denoising Diffusion Probabilistic Models}
\label{subsection:DDPM}
Let $\mu \in \mathscr{P}(\bbR^d)$ denote the target data distribution from which we aim to generate samples. Given a collection of samples drawn from $\mu$, the forward process of \acp{ddpm} progressively corrupts the data into Gaussian noise. This transformation is characterized by a time-rescaled \ac{ou} process, governed by the following \ac{sde}
\begin{equation}
	\d X^\OU_t = -\gamma(t)X^{\OU}_t \d t + \sqrt{2\gamma(t)} \d W_t, \quad X^\OU_0 \sim \mu.
	\label{eq:OU process}
\end{equation}
Here, $W_t$ denotes a $d$-dimensional standard Wiener process, and the positive-valued function $\gamma\colon \bbR_+ \to \bbR_+$ serves as the time-rescaling factor. A key property of the \ac{ou} process is that it converges exponentially fast to its stationary distribution, the standard $d$-dimensional Gaussian distribution. Specifically, $X^\OU_t$ is distributed identically to $e^{-\Gamma(t)}X^\OU_0 + \sqrt{1 - e^{-2\Gamma(t)}} \, \Xi^d$, where $\Xi^d \sim \xi^d$ and $\Gamma(t) \coloneqq \int_{0}^{t}\gamma(\tau) \d \tau$ represents the effective running time of the \ac{ou} process.

Next, the generative phase of \acp{ddpm} simulates the time-reversal of this forward \ac{ou} process. Let $X_t$ denote the reverse process defined over the fixed time interval $[0,T]$, governed by the following \ac{sde}
\begin{equation}
	\d X_t = \bar{\gamma}(t)\left(X_t + 2\nabla \log\rho_t(X_t)\right)\d t + \sqrt{2\bar{\gamma}(t)} \d W_t,\footnote{For convenience, we do not distinguish between the forward and reverse Wiener processes, using the same notation $W$ for both.}  \quad X_0 \sim \rho_0.
	\label{eq:reversed OU process:simplified notation}
\end{equation}
Here, $\bar{\gamma}(t) = \gamma(T-t)$ represents the reversed time-rescaling function, and $\rho_t$ denotes the distribution of $X_t$, which coincides with the marginal distribution of the forward \ac{ou} process at time $T-t$ (\ie, $X^\OU_{T-t} \sim \rho_{t}$). 

In practice, simulating the exact reverse process is intractable. Therefore, practical \acp{ddpm} are governed by an approximate reverse \ac{sde}
\begin{equation}
		\d \widehat{X}_t = \bar{\gamma}(t)\big(\widehat{X}_t + 2S_{t_k}\big(\widehat{X}_{t_k}\big)\big)\d t + \sqrt{2\bar{\gamma}(t)}\d W_t, 
		\quad \widehat{X}_0 \sim \xi^d, \; t \in [t_k, t_{k+1}), \; k = 0, 1, \dots, K-1.
	\label{eq:DDPM:SDE form}
\end{equation}
Compared to the exact reverse process in~\eqref{eq:reversed OU process:simplified notation}, this formulation introduces three key distinctions.
\begin{enumerate*}
	\item The exact initial distribution $\rho_0$ is replaced by the standard Gaussian distribution $\xi^d$, as $\rho_0$ is intractable to evaluate or sample from directly.
	\item The true score function $\nabla\log \rho_t$ is replaced by an approximate score function $S\colon [0,T]\times \bbR^d\to\bbR^d$, typically learned via score-matching techniques~\cite{JMLR:v6:hyvarinen05a, 6795935}. 
	\item A time discretization scheme $0 = t_0 < t_1 < \cdots < t_K = T$ is applied, as exact continuous integration along the score function is computationally infeasible.
\end{enumerate*}

These three practical modifications inherently introduce three corresponding sources of error in the generative phase.
\begin{enumerate*}
	\item \emph{Mixing error}, arising from the mismatch between $\rho_0$ and the standard Gaussian (\ie, the failure of the forward \ac{ou} process to perfectly reach its stationary distribution in finite time $T$).
	\item \emph{Score estimation error}, stemming from the approximation gap between the learned surrogate $S$ and the true score $\nabla\log \rho$.
	\item \emph{Discretization error}, caused by holding the score function evaluation constant (\ie, using $S_{t_k}$) over each time interval $t\in[t_k,t_{k+1})$. 
\end{enumerate*} 
In the sequel, these errors collectively determine the overall \Renyi divergence error bound for \acp{ddpm}, as characterized in Theorem~\ref{theorem:Renyi divergence error of DDPM}.

Finally, as a side remark, the standard iteration rule for \acp{ddpm} is derived from~\eqref{eq:DDPM:SDE form} using an exponential integrator, however, we adopt the \ac{sde} form in~\eqref{eq:DDPM:SDE form} because it is more convenient for theoretical analysis. We defer further details regarding \acp{ddpm} to Appendix~\ref{section:appendix:backgrounds on linear SDE}.

\subsection{\Renyi Divergence}
\label{subsection:Renyi divergence}
This section formally defines the \Renyi divergence and then provides a quantitative characterization of its enhanced sensitivity.
\begin{definition}[\Renyi Divergence]
	For $\alpha \in (0, +\infty]$ and probability distributions $P$ and $Q$, the $\alpha$-\Renyi divergence of $P$ relative to $Q$ is defined as
	\begin{equation}
	\RenyiDivergence{P}{Q}\coloneqq
	\begin{cases}
		\displaystyle 
		\frac{1}{\alpha-1}
		\log \mathbb{E}_{Q}\left[
		\left(\frac{\mathrm d P}{\mathrm d Q}\right)^\alpha
		\right],
		& \alpha\in(0,1)\cup(1,+\infty),\\[1.2em]
		\displaystyle 
		\KLDivergence{P}{Q}
		=\mathbb{E}_{P}\left[
		\log\frac{\mathrm dP}{\mathrm dQ}
		\right],
		& \alpha=1,\\[1.2em] 
		\displaystyle 
		\operatorname*{ess\,sup}_{Q}\log
		\frac{\mathrm dP}{\mathrm dQ},
		& \alpha=\infty.
	\end{cases}
	\end{equation}
	\label{definition:Renyi divergence}
\end{definition}
We conclude this section by formally demonstrating the enhanced sensitivity of the \Renyi divergence in low-probability regions, as motivated in Section~\ref{section:Introduction}. The following proposition establishes that two probability distributions that are close in \ac{tv}, \ac{kl} divergence, or $p$-Wasserstein distance can nevertheless remain far apart under the \Renyi divergence. 
We provide an intuitive explanation in Remark~\ref{remark:when large Renyi divergence occurs} and defer the complete proof in Appendix~\ref{section:Appendix:Why Renyi}.

\begin{proposition}
	\label{proposition:KL-close distributions can diverge under Renyi divergence}
	For any $\alpha>1$, $p>1$, and $0<\eta<M$, there exist probability distributions $P$ and $Q$ satisfying the following conditions simultaneously.
	\begin{enumerate}
		\item $\max\left\{\totalVariation{P}{Q}, \KLDivergence{P}{Q}, \WassersteinDistance[p]{P}{Q} \right\} < \eta$.
		\item $\RenyiDivergence{P}{Q} > M$.
	\end{enumerate}
\end{proposition}

\begin{remark}
	As an immediate corollary, there exists no constant $C(\alpha)$ (depending only on $\alpha$) such that the inequality $\RenyiDivergence{P}{Q} \le C(\alpha)D(P,Q)$ holds uniformly across all pairs of distributions $(P,Q)$, where $D(P,Q)$ represents any of \ac{tv}, \ac{kl} divergence, or $p$-Wasserstein distance.
\end{remark}

\begin{remark}
	\label{remark:when large Renyi divergence occurs}
	Intuitively, this phenomenon occurs due to how different discrepancy measures penalize distributional mismatches in tail regions. Consider a set $A$ in the tail of $P$ with very small mass $P(A) \ll 1$. Suppose $Q$ allocates substantially less mass to $A$, making the density ratio $\d P/\d Q$ significantly larger than $1$ on $A$. On the complement $A^c$, $Q$ is constructed to match $P$ closely (\eg, by setting $\d P/\d Q \approx 1$).
	For \ac{tv}, \ac{kl} divergence, and $p$-Wasserstein distance, the vanishing mass $P(A)$ dominates and suppresses the large ratio $\d P/\d Q$ over $A$. Coupled with the near-identical behavior on $A^c$, these discrepancy measures remain small ($< \eta$). In contrast, for $\alpha>1$, the exponentiation $(\d P/\d Q)^\alpha$  in the \Renyi divergence  overwhelms the small probability mass $P(A)$, resulting in a large divergence ($> M$).
\end{remark}

\section{Main Results}
\label{section:Main Results}

In this section, we first consider two general diffusion processes with mismatched drifts, establishing upper bounds on their \Renyi divergence in Section~\ref{subsection:Renyi divergence for diffusion processes with mismatched drifts}. Because these results apply broadly beyond \acp{ddpm}, they are presented in a general framework. In Section~\ref{subsection:Renyi divergence error for DDPM}, we specialize these findings to \acp{ddpm}, characterizing their convergence under \Renyi divergence in terms of both distributional error and iteration complexity. Next, we provide a technical overview of the proof in Section~\ref{subsection:technique overview}, deferring full proofs to Appendix~\ref{section:Appendix:proof of main results}. Finally, in Section~\ref{subsection:Renyi divergence error for DDPM:CD}, we extend our analysis to their kinetic counterpart, a class of \acp{ddpm} governed by critically damped Langevin dynamics.

\subsection{\Renyi Divergence for Diffusion Processes with Mismatched Drifts}
\label{subsection:Renyi divergence for diffusion processes with mismatched drifts}

Let $Z_t$ and $\widehat{Z}_t$ be two diffusion processes that share the same diffusion coefficient and initial distribution, yet differ in their drift coefficients. Suppose they are governed by the following \acp{sde}
\begin{align}
	\d Z_t &= b_t(Z_t)\d t + \sigma_t\d W_t, 
	\label{eq:target diffusion} \\
	\d \widehat{Z}_t &= \widehat{b}_t(\widehat{Z}_t)\d t + \sigma_t\d W_t, 
	\label{eq:learned diffusion}
\end{align}
where $\sigma_t \in \bbR$, and $b_t, \widehat{b}_t\colon \bbR^d \to \bbR^d$ are Lipschitz continuous and continuously differentiable vector fields. 

The following theorem establishes upper bounds on the \Renyi divergence between $Z_t$ and $\widehat{Z}_t$. We derive these results via a Girsanov argument, with detailed proofs deferred to Appendix~\ref{section:Appendix:proof of theorem:Renyi divergence bound}.

\begin{theorem}
	\label{theorem:Renyi divergence bound}
	Take $Z$ as the reference process and denote by $\rho$ its path measure. Let $\rho_t$ and $\widehat{\rho}_t$ denote the marginal distributions of $Z_t$ and $\widehat{Z}_t$, respectively. 
	For any $\alpha > 1$ and any tunable parameter $\beta > 1$ with Hölder conjugate exponent $\beta^*$ (\ie, $1/\beta + 1/\beta^* = 1$), the following upper bounds hold.
	\begin{enumerate}
		\item The \Renyi divergence of $\widehat{\rho}_t$ relative to $\rho_t$ satisfies
		\begin{equation}
			\RenyiDivergence{\widehat{\rho}_t}{\rho_t} 
			\leq \frac{1}{(\alpha - 1)\beta^*}\log
			\expectation{\rho}{\exponential{\frac{\alpha\beta^*(\alpha\beta-1)}{2}\int_0^t\frac{1}{\sigma_s^2}\left\|b_s(Z_s)-\widehat{b}_s(Z_s)\right\|^2\d s}}.
			\label{eq:Renyi divergence bound:quadratic variation}
		\end{equation}
	
		\item Provided that Novikov's condition
		$ \displaystyle \expectation{\rho}{\exponential{\frac{1}{2}\int_0^t\frac{1}{\sigma_s^2}\left\|b_s(Z_s)-\widehat{b}_s(Z_s)\right\|^2\d s}} < \infty$ holds, the \Renyi divergence of $\rho_t$ relative to $\widehat{\rho}_t$ satisfies
		\begin{equation}
			\RenyiDivergence{\rho_t}{\widehat{\rho}_t}
			\leq \frac{1}{(\alpha-1)\beta^*}\log
			\expectation{\rho}{\exponential{\frac{(\alpha-1)\beta^*((\alpha-1)\beta+1)}{2}\int_0^t\frac{1}{\sigma_s^2}\left\|b_s(Z_s)-\widehat{b}_s(Z_s)\right\|^2\d s}}.
			\label{eq:Renyi divergence bound between Q and Q-hat}
		\end{equation}
	\end{enumerate}
\end{theorem}

\begin{remark}
	These bounds are characterized by the \ac{cgf} of the ``drift error'', which is expressed as a path integral of the squared drift difference along the trajectory of the process $X$. In particular, when $\widehat{b}_\tau = b_\tau$ for all $\tau \in [0,t]$, both bounds yield $\RenyiDivergence{\widehat{\rho}_t}{\rho_t} = \RenyiDivergence{\rho_t}{\widehat{\rho}_t} = 0$, as expected. 
\end{remark}

\subsection{\Renyi Divergence Error of \ac{ddpm}}
\label{subsection:Renyi divergence error for DDPM}

In this subsection, we apply Theorem~\ref{theorem:Renyi divergence bound} to characterize the convergence of \acp{ddpm} under the \Renyi divergence. To this end, we first formally define the score estimation error and state our assumptions on the data distribution and the regularity of the score function.

\begin{definition}[Score Estimation Error]
	\label{definition:score estimation error}
	Given a sample path $X$ and a time discretization $0=t_0<t_1<\cdots< t_K = T$, the pathwise score estimation error $\scoreError(X)$ is defined as the weighted average of the squared score error along $X$, given by
	\begin{equation}
		\scoreError(X) \coloneqq \sum_{k=0}^{K-1} \frac{\Delta\tilde{\Gamma}_{k}}{\tilde{\Gamma}(T)} \|S_{t_k}(X_{t_{k}}) - \nabla\log \rho_{t_k}(X_{t_k})\|^2.
	\end{equation}
	Here, the weights $\displaystyle  \frac{\Delta\tilde{\Gamma}_{k}}{\tilde{\Gamma}(T)}$ normalize the effective step sizes relative to the total effective time under the time-rescaling function $\bar{\gamma}$ of the reverse process~\eqref{eq:reversed OU process:simplified notation}, where
	\begin{equation}
		\tilde{\Gamma}(t) \coloneqq \int_{0}^{t}\bar{\gamma}(\tau)\d \tau, 
		\quad 
		\Delta\tilde{\Gamma}_{k} \coloneqq \tilde{\Gamma}(t_{k+1}) - \tilde{\Gamma}(t_k) = \int_{t_k}^{t_{k+1}}\bar{\gamma}(\tau)\d \tau.
	\end{equation}
	Note that $\displaystyle \sum_{k=0}^{K-1} \frac{\Delta\tilde{\Gamma}_{k}}{\tilde{\Gamma}(T)} = 1$, making $\scoreError(X)$ a convex combination of step-wise score errors.
\end{definition}

\begin{definition}[Maximum Effective Discretization Step]
	\label{definition:maximum effective discretization step}
	The maximum effective discretization step $\Delta_{\max}$ is defined as 
	\begin{equation}
		\Delta_{\max} \coloneqq \max_{0 \le k \le K-1} \Delta\tilde{\Gamma}_k = \max_{0 \le k \le K-1}\int_{t_k}^{t_{k+1}}\bar{\gamma}(\tau)\d \tau.
	\end{equation}
\end{definition}

\begin{assumption}[Data Distribution]
	\label{assumption:data distribution}
	The target data distribution $\mu \in \mathscr{P}(\bbR^d)$ is sub-Gaussian with parameter $\sigma > 0$. Mathematically,
	\begin{equation}
		\expectation{Z\sim \mu}{\exponential{\lambda\innerProd{v}{Z-\expectation{}{Z}}}}\leq \exp\left(\frac{1}{2}\lambda^2\sigma^2\right),
	\end{equation}
	holds for all $\lambda \in \bbR$ and $v \in \mathbb{S}^{d-1}$, where $\mathbb{S}^{d-1}$ denotes the unit sphere in $\bbR^d$.
\end{assumption}

\begin{assumption}[Lipschitz Score Function]
	\label{assumption:Lipschitz score}
	The score function $\nabla\log \rho_t(\cdot)$ is $L$-Lipschitz continuous on $\bbR^d$ uniformly for all $t \in [0,T]$. Specifically, there exists a constant $L > 0$ such that
	\begin{equation}
		\left\| \nabla\log \rho_t(x) - \nabla\log \rho_t(y) \right\| \le L \|x - y\|, \quad \forall x, y \in \bbR^d, \; \forall t \in [0,T].
	\end{equation}
\end{assumption}

Now we are equipped to characterize the error of \acp{ddpm} under \Renyi divergence, as shown in the following theorem.

\begin{theorem}[\Renyi Divergence Error of \acp{ddpm}]
	\label{theorem:Renyi divergence error of DDPM}
	Suppose Assumptions~\ref{assumption:data distribution} and \ref{assumption:Lipschitz score} hold. Let $\rho$ denote the path measure of the exact reverse process $X$ in~\eqref{eq:reversed OU process:simplified notation}. For any order $\alpha > 1$, the \Renyi divergence of the target distribution $\mu$ relative to the \ac{ddpm}-generated distribution $\widehat{\mu}$ satisfies
	\begin{equation}
			\RenyiDivergence{\mu}{\widehat{\mu}} \lesssim \cE^{M} + \cE^{S} + \cE^{D} + o\left(\cE^{D}\right),
			\label{eq:Renyi divergence error of DDPM}
	\end{equation}
	where $o\left(\cE^{D}\right)$ denotes a higher order term satisfying $o\left(\cE^{D}\right)/\cE^{D} \to 0$ as $\Delta_{\max} \to 0$. The mixing error $\cE^{M}$, score estimation error $\cE^{S}$, and discretization error $\cE^{D}$ are defined by
	\begin{align}
		\cE^{M} &\coloneqq \exponential{-\frac{2\tilde{\Gamma}(T)}{2\alpha-1}}\RenyiDivergence[2\alpha-1]{\mu}{\xi^d},	\label{eq:mixing error}\\
		\cE^{S} &\coloneqq \frac{1}{\alpha-1}\log\expectation{\rho}{\exponential{8(2\alpha-1)(4\alpha-1)\tilde{\Gamma}(T)\scoreError(X)}}, \label{eq:score error}\\ 
		\cE^{D} &\coloneqq \frac{\alpha^2}{\alpha-1} L^2\tilde{\Gamma}(T) \left( d \Delta_{\max} + \left(m^2+\sigma^2d\right)\Delta_{\max}^2 \right). \label{eq:discretization error}
	\end{align}
	Here, $m \coloneqq \|\expectation{Z\sim \mu}{Z}\|$ in~\eqref{eq:discretization error} denotes the magnitude of the target distribution mean.
\end{theorem}

\begin{remark}
	While the underlying mechanisms of $\cE^M$, $\cE^S$, and $\cE^D$ have been established in Section~\ref{subsection:DDPM}, their quantitative scaling under \Renyi divergence warrants special emphasis.
	\begin{enumerate*}
		\item The mixing error $\cE^M$ decays exponentially fast with exponent $\cO(\widetilde{\Gamma}(T) / \alpha)$, in contrast to the exponent $\cO(\widetilde{\Gamma}(T))$ established for the \ac{tv} distance and \ac{kl} divergence.
		\item The score estimation error $\cE^S$ is governed by the \ac{cgf} of the pathwise squared error $\scoreError(X)$, imposing a more stringent requirement than the standard $L^2$-error control.
		\item The discretization error $\cE^D$ scales linearly in the maximum effective discretization step $\Delta_{\max}$, with a coefficient whose dependence on $\alpha$ is $\cO(\alpha)$. This follows because $\sigma$ scales as $\cO(1)$ and $m$ scales as $\cO(\sqrt{d})$ in the high-dimensional regime.
	\end{enumerate*}
\end{remark}

\begin{remark}[Trade-Off in Effective Running Time $\widetilde{\Gamma}(T)$]
	The upper bound in~\eqref{eq:Renyi divergence error of DDPM} reveals a fundamental trade-off with respect to the total effective time $\widetilde{\Gamma}(T)$. On one hand, the mixing error $\cE^{\mathrm{M}}$ decays exponentially with respect to $\widetilde{\Gamma}(T)$. On the other hand, a larger $\widetilde{\Gamma}(T)$ linearly inflates the discretization error $\cE^{D}$ and amplifies the score estimation error $\cE^{S}$. Crucially, because $\cE^{S}$ takes the form of a \ac{cgf} of $\scoreError$, its growth rate with respect to $\widetilde{\Gamma}(T)$ is dictated by the tail behavior of the score error under the path measure $\rho$. In particular, if $\scoreError$ is sub-exponential or sub-Gaussian under $\rho$, $\cE^{S}$ grows linearly (within a neighborhood of 0) or quadratically in $\widetilde{\Gamma}(T)$, respectively.
\end{remark}

\begin{remark}[Comparison with \ac{tv} and \ac{kl} Bounds]
	Theorem~\ref{theorem:Renyi divergence error of DDPM} serves as a \Renyi divergence counterpart to the \ac{tv} error bound established in~\cite[Theorem 2]{chen2022sampling} and the \ac{kl} error bound established in~\cite[Theorem 1]{chen2023improved}. When adapted to our time-rescaling framework, their respective results can be expressed as
	\begin{align}
		\squaredTotalVariation{\mu}{\widehat{\mu}} &\lesssim e^{-2\widetilde{\Gamma}(T)}\KLDivergence{\mu}{\xi^d} + \expectation{\rho}{\scoreError(X)}\widetilde{\Gamma}(T) + L^2\widetilde{\Gamma}(T) \left(d\Delta_{\max} + m_2 \Delta_{\max}^2\right),\\
		\KLDivergence{\mu}{\widehat{\mu}} &\lesssim e^{-\widetilde{\Gamma}(T)}\left(m_2 + d\right) + \expectation{\rho}{\scoreError(X)}\widetilde{\Gamma}(T) 
		+ L^2\widetilde{\Gamma}(T)^2d\Delta_{\max},
	\end{align}
	where $m_2 \coloneqq \expectation{Z\sim \mu}{\|Z\|^2}$ denotes the second moment of the target distribution $\mu$. Note that the three terms on the right-hand side of each bound share the same interpretations, corresponding to mixing, score estimation, and discretization errors.
\end{remark}

Theorem~\ref{theorem:Renyi divergence error of DDPM} yields the following iteration complexity for \acp{ddpm}. The proof is deferred to Appendix~\ref{subsection:proof of corollary:iteration complexity of DDPM}.

\begin{corollary}[Iteration Complexity of \acp{ddpm}]
	\label{corollary:iteration complexity of DDPM}
	Let $\epsilon \in \left(0, 1\right)$ be the target error tolerance, and assume the initial \Renyi divergence $R_0 \coloneqq \RenyiDivergence[2\alpha-1]{\mu}{\xi^d}$ is finite. Define the logarithmic factor $\mathcal{L}_\epsilon \coloneqq \log \left(1 + R_0 / \epsilon\right)$ and configure the \ac{ddpm} parameters satisfying $\widetilde{\Gamma}(T) = C\alpha \mathcal{L}_\epsilon$ for some constant $C>1$ and $\displaystyle \Delta_{\max} \asymp \frac{(\alpha-1)\epsilon}{\alpha^3 L^2 d \cdot \mathcal{L}_\epsilon} \wedge \frac{d}{m^2} \wedge \frac{1}{\sigma^2} \wedge \frac{1}{L}$.
	If the score estimation error satisfies $\cE^{S} \lesssim \epsilon$ under these settings, then the output distribution $\widehat{\mu}$ satisfies $\RenyiDivergence{\mu}{\widehat{\mu}} \lesssim \epsilon$. Moreover, this guarantee is achieved within $K$ iterations, where
	\begin{equation}
		K \asymp \frac{\alpha^4 L^2 d}{(\alpha-1)\epsilon} \mathcal{L}_\epsilon^2 \vee \alpha\left(\sigma^2 \vee \frac{m^2}{d} \vee L\right) \mathcal{L}_\epsilon.
		\label{eq:iteration complexity of DDPM}
	\end{equation}
\end{corollary}

\begin{remark}[Comparison with \ac{tv} and \ac{kl} Iteration Complexity]
	\label{remark:iteration complexity of DDPM}
	In the small-$\epsilon$ regime, the first term in~\eqref{eq:iteration complexity of DDPM} dominates, yielding an asymptotic iteration complexity of $K = \tilde{\Theta}\left(\alpha^3 L^2 d / \epsilon\right)$ for large $\alpha$. 
	This recovers the $\tilde{\Theta}\left(L^2 d / \epsilon\right)$ iteration complexity required for $\squaredTotalVariation{\mu}{\widehat{\mu}} \lesssim \epsilon$ in~\cite[Theorem 2]{chen2022sampling} and for $\KLDivergence{\mu}{\widehat{\mu}} \lesssim \epsilon$ in~\cite[Theorem 1]{chen2023improved}.
	Furthermore, provided sufficiently accurate score estimation, achieving the stronger guarantee $\RenyiDivergence{\mu}{\widehat{\mu}} \lesssim \epsilon$ for large $\alpha$ requires at most an $\widetilde{\mathcal{O}}\left(\alpha^3\right)$ factor increase in iterations compared to ensuring either $\squaredTotalVariation{\mu}{\widehat{\mu}} \lesssim \epsilon$ or $\KLDivergence{\mu}{\widehat{\mu}} \lesssim \epsilon$.
\end{remark}

As a direct application of Theorem~\ref{theorem:Renyi divergence error of DDPM} and Corollary~\ref{corollary:iteration complexity of DDPM}, we obtain the following error bound and iteration complexity under a uniform $L^\infty$-score estimation.

\begin{corollary}[\Renyi Divergence Error under $L^\infty$-Score Estimation]
	\label{corollary:Renyi divergence error for DDPM:L-infinity score estimation}
	Under Assumptions~\ref{assumption:data distribution} and \ref{assumption:Lipschitz score}, suppose the score estimator satisfies a uniform $L^\infty$-error bound $\displaystyle \max_{0 \le k \le K-1}\|S_{t_k}(\cdot) - \nabla\log \rho_{t_k}(\cdot)\|_\infty \leq \epsilon_{\infty}$. Then, the \Renyi divergence satisfies
	\begin{equation}
		\RenyiDivergence{\mu}{\widehat{\mu}} \lesssim \cE^{M} + \frac{\alpha^2}{\alpha-1}\widetilde{\Gamma}(T)\epsilon_{\infty}^2 + \cE^{D} + o\left(\cE^{D}\right).
		\label{eq:Renyi divergence error under L-infty error}
	\end{equation}
	Moreover, for any $\epsilon \in \left(0, \dfrac{1}{\alpha-1}\right)$, if the score error tolerance satisfies $ \displaystyle \epsilon_{\infty} \lesssim \sqrt{\frac{(\alpha-1)\epsilon}{\alpha^3 \mathcal{L}_\epsilon}}$,
	where $\mathcal{L}_\epsilon = \log\left(1 + R_0 / \epsilon\right)$ is defined as in Corollary~\ref{corollary:iteration complexity of DDPM}, then the generated distribution satisfies $\RenyiDivergence{\mu}{\widehat{\mu}} \lesssim \epsilon$ with the iteration complexity $K$ given in~\eqref{eq:iteration complexity of DDPM}.
\end{corollary}

\begin{proof}
	Since $\log\expectation{}{\exponential{\lambda\scoreError}} \leq \lambda \epsilon_{\infty}^2$ for any $\lambda > 0$, \eqref{eq:Renyi divergence error under L-infty error} is an immediate consequence of Theorem~\ref{theorem:Renyi divergence error of DDPM}. The condition $ \displaystyle \epsilon_{\infty} \lesssim \sqrt{\frac{(\alpha-1)\epsilon}{\alpha^3 \mathcal{L}_\epsilon}}$ follows from $\cE^S\lesssim \epsilon $ as required in Corollary~\ref{corollary:iteration complexity of DDPM}.
\end{proof}

\begin{remark}[Comparison under $L^\infty$-Score Estimation]
	\label{remark:compare results under L-infinity score estimation error}
	\lwl{A related upper bound on $\RenyiDivergence{\widehat{\mu}}{\mu}$ was established in~\cite[Theorem 3]{yang2022convergence} under an $L^\infty$-bounded score error. Their result was established for inexact \ac{lmc} under the assumption that the target distribution satisfies an \ac{lsi}, a different setting from ours as mentioned in Section~\ref{section:Related Work}. Nonetheless, the resulting iteration complexity shares a similar scaling. Specifically, they showed that if $\epsilon_\infty \lesssim \sqrt{\epsilon/\alpha^2}$, then $\RenyiDivergence{\widehat{\mu}}{\mu} \lesssim \epsilon$ with an iteration complexity of $K \asymp \alpha^3/\epsilon$. Up to poly-logarithmic factors, their result is consistent with Corollary~\ref{corollary:Renyi divergence error for DDPM:L-infinity score estimation} for large $\alpha$.}
\end{remark}

\subsection{Technique Overview}
\label{subsection:technique overview}

This subsection provides an overview of the proof of Theorem~\ref{theorem:Renyi divergence error of DDPM}, with complete technical details deferred to Appendix~\ref{section:Appendix:proof of main results}.

We construct an auxiliary process $X^{\mathrm{Aux}}_t$ initialized at $X^{\mathrm{Aux}}_0 \stackrel{\d}{=} X_0$ that evolves according to the \ac{sde}
\begin{equation}
	\d X^\Aux_t = \bar{\gamma}(t)\left(X^\Aux_t + 2S_{t_k}(X^\Aux_{t_k})\right)\d t + \sqrt{2\bar{\gamma}(t)}\d W_t,
	\label{eq:auxillary process X-hat}
\end{equation}
for $t\in [t_k,t_{k+1})$ and $k\in\{0,1,\dots,K-1\}$. 
This auxiliary construction bridges the initial distribution mismatch between the practical process $\widehat{X}$ and the exact reverse process $X$ (\cf~\eqref{eq:reversed OU process:simplified notation} and \eqref{eq:DDPM:SDE form}), enabling the application of Theorem~\ref{theorem:Renyi divergence bound}.
By construction, $X^\Aux_t$ and $\widehat{X}_t$ evolve along the same \ac{sde}, differing solely in their initial distributions ($X^\Aux_0 \sim \rho_0$ versus $\widehat{X}_0 \sim \xi^d$).
Denoting by $\rho_t$ and $\widehat{\rho}_t$ respectively the marginal distribution of $X$ and $\widehat{X}$ at time $t$, we have $\mu = \rho_T$ and $\widehat{\mu} = \widehat{\rho}_T$. Invoking the weak triangle inequality for \Renyi divergence (\cf~Lemma~\ref{lemma:weak triangle inequality}) yields the decomposition
\begin{equation}
	\RenyiDivergence{\mu}{\widehat{\mu}} 
	= \RenyiDivergence{\rho_T}{\widehat{\rho}_T}
	\leq \frac{2\alpha-1}{2(\alpha-1)}\RenyiDivergence[2\alpha]{\rho_T}{\rho^\Aux_T} + \RenyiDivergence[2\alpha-1]{\rho^\Aux_T}{\widehat{\rho}_T},
	\label{eq:decompose the Renyi divergence into two parts}
\end{equation}
where $\rho^\Aux_t$ denotes the law of $X^\Aux_t$. Since $\rho^\Aux_T$ and $\widehat{\rho}_T$ are the outputs of the same Markov channel characterized by \ac{sde}~\eqref{eq:auxillary process X-hat} with inputs $\rho_0$ and $\xi^d$, respectively, the data processing inequality (\cf~Lemma~\ref{lemma:data processing inequality of Renyi divergence}) yields
\begin{equation}
	\RenyiDivergence[2\alpha-1]{\rho^\Aux_T}{\widehat{\rho}_T}
	\leq \RenyiDivergence[2\alpha-1]{\rho_0}{\xi^d} \\
	\leq \exponential{-\frac{2\widetilde{\Gamma}(T)}{2\alpha-1}}\RenyiDivergence[2\alpha-1]{\mu}{\xi^d},
	\label{eq:proof of main result:intermediate result 1}
\end{equation}
where the second inequality follows from $X^\OU_T \sim \rho_0$ together with the exponential decay of \Renyi divergence along the forward \ac{ou} process (\cf Lemma~\ref{lemma:convergence of Renyi divergence along LD under LSI}). 

Next, we apply Theorem~\ref{theorem:Renyi divergence bound} to bound the remaining term $\frac{2\alpha-1}{2(\alpha-1)}\RenyiDivergence[2\alpha]{\rho_T}{\rho^\Aux_T}$ on the right-hand side of \eqref{eq:decompose the Renyi divergence into two parts}, which yields
\begin{equation}
	\frac{2\alpha-1}{2(\alpha-1)}\RenyiDivergence[2\alpha]{\rho_T}{\rho^\Aux_T} \leq \frac{1}{2(\alpha-1)\beta^*}\Phi\left((2\alpha-1)\beta^*\left((2\alpha-1)\beta + 1 \right)\right),
	\label{eq:express the upper-bound in terms of Phi lambda}
\end{equation}
where we define the function $\Phi$ by
\begin{equation}
	\Phi(\lambda)\coloneqq
	\log\expectation{X\sim\rho}{\exponential{\lambda \sum_{k=0}^{K-1} \int_{t_k}^{t_{k+1}}\bar{\gamma}(t)\left\|S_{t_k}\left(X_{t_k}\right) - \nabla\log \rho_{t}(X_t)\right\|^2 \d t}}.
	\label{eq:an expression that we want to upper-bound}
\end{equation}
Thus, it remains to establish an upper bound on $\Phi(\lambda)$, which is provided by the following proposition.
\begin{proposition}
	\label{proposition:upper-bound of Phi(lambda)}
	Let $\Phi(\lambda)$ be defined as in~\eqref{eq:an expression that we want to upper-bound}. We have
	\begin{equation}
		\Phi(\lambda)\lesssim \log\expectation{\rho}{\exponential{4\lambda\tilde{\Gamma}(T)\scoreError(X)}} +  (1+o(1))\underbrace{\lambda L^2\widetilde{\Gamma}(T)
		\left( d \Delta_{\max}  + \left(m^2+\sigma^2d\right)\Delta_{\max} ^2\right)}_{\textrm{denoted by }\cE(\lambda)},
		\label{eq:upper-bound the target in terms of lambda:proposition}
	\end{equation}
	where $o(1)$ denotes a term that tends to 0 as $\cE(\lambda) \to 0$.
\end{proposition}
The proof of Proposition~\ref{proposition:upper-bound of Phi(lambda)} relies on a score perturbation inequality (\cf~Lemma~\ref{lemma:score perturbation}), a score concentration lemma (\cf~Lemma~\ref{lemma:score concentration}), and an upper bound on the moment generating function of quadratic sub-Gaussian random vectors (\cf~Lemma~\ref{lemma:moments of quadratic sub-Gaussian}). The complete proof is deferred to Appendix~\ref{section:Appendix:proof of main results}.

With Proposition~\ref{proposition:upper-bound of Phi(lambda)} established, we set $\beta = \beta^* = 2$ and substitute $\lambda = 2(2\alpha-1)(4\alpha-1) \lesssim \alpha^2$ into \eqref{eq:upper-bound the target in terms of lambda:proposition}. Combining \eqref{eq:decompose the Renyi divergence into two parts}, \eqref{eq:proof of main result:intermediate result 1}, and \eqref{eq:express the upper-bound in terms of Phi lambda} completes the proof of Theorem~\ref{theorem:Renyi divergence error of DDPM}.

\subsection{Extension to Critically Damped \acp{ddpm}}
\label{subsection:Renyi divergence error for DDPM:CD}
In this subsection, we extend our \Renyi divergence convergence analysis of \acp{ddpm} to their kinetic variants, which are governed by critically damped Langevin dynamics. Compared to standard \acp{ddpm} that rely on a standard \ac{ou} process, critically damped \acp{ddpm}~\cite{dockhorn2021score} augment the state space with an auxiliary velocity (or momentum) variable to form a kinetic \ac{ou} process, governed by the \acp{sde}
\begin{equation}
	\left\{
	\begin{aligned}
		\d X^\CD_t &= \gamma(t)V^\CD_t \d t, & X^\CD_0 &\sim \mu, \\
		\d V^\CD_t &= -\gamma(t)\left(X^\CD_t + 2 V^\CD_t\right)\d t + 2\sqrt{\gamma(t)} \d W_t, & V^\CD_0 &\sim \xi^d.
	\end{aligned}
	\right.
	\label{eq:forward_CD_OU}
\end{equation}
Let $[X_t, V_t]^\top$ denote the time-reversal of $[X^\CD_t, V^\CD_t]^\top$ over the time interval $[0, T]$. Analogous to~\eqref{eq:reversed OU process:simplified notation}, $[X_t, V_t]^\top$ satisfies the reverse \acp{sde}
\begin{equation}
	\left\{
	\begin{aligned}
		\d X_t &= -\bar{\gamma}(t) V_t\d t, & X_0 &\stackrel{\d}{=} X^\CD_T, \\
		\d V_t &= \bar{\gamma}(t) \left(X_t + 2V_t + 4\nabla_{V_t}\log \rho_t \left(X_t, V_t\right)\right)\d t + 2\sqrt{\bar{\gamma}(t)}\d W_t, & V_0 &\stackrel{\d}{=} V^\CD_T.
	\end{aligned}
	\right.
	\label{eq:exact_reverse_CD}
\end{equation}
Here, with a slight abuse of notation, we use $\rho_t$ to denote the marginal distribution of $[X_t, V_t]^\top$.

Analogous to the standard setting in~\eqref{eq:DDPM:SDE form}, the practical algorithmic implementation of critically damped \ac{ddpm} replaces the true score $\nabla_{V_t}\log \rho_t$ with a learned score estimator $S^\CD: \mathbb{R}_+ \times \mathbb{R}^d \times \mathbb{R}^d \to \mathbb{R}^d$ evaluated at discrete time steps $0 = t_0 < t_1 < \dots < t_K = T$. Over each interval $t \in [t_k, t_{k+1})$, the practical reverse process $[\widehat{X}_t, \widehat{V}_t]^\top$ evolves according to \acp{sde}
\begin{equation}
	\left\{
	\begin{aligned}
		\d \widehat{X}_t &= -\bar{\gamma}(t)\widehat{V}_t \d t, & \widehat{X}_0 &\sim  \xi^{d}, \\
		\d \widehat{V}_t &= \bar{\gamma}(t) \left(\widehat{X}_t + 2\widehat{V}_t + 4S^\CD_{t_k} \big(\widehat{X}_{t_k}, \widehat{V}_{t_k}\big)\right)\d t + 2\sqrt{\bar{\gamma}(t)}\d W_t, &  \widehat{V}_0  &\sim \xi^{d}.
	\end{aligned}
	\right.
	\label{eq:practical_reverse_CD}
\end{equation}

To analyze the \Renyi divergence error for critically damped \acp{ddpm}, we adapt Definition~\ref{definition:score estimation error} and Assumption~\ref{assumption:Lipschitz score} as follows.

\begin{definition}[Score Error for Critically Damped \ac{ddpm}]
	\label{definition:score estimation error:CD}
	The pathwise score estimation error of critically damped \ac{ddpm} is defined as
	\begin{equation}
		\scoreErrorCD(X,V) \coloneqq \sum_{k=0}^{K-1}\frac{\Delta\tilde{\Gamma}_{k}}{\tilde{\Gamma}(T)} \left\|S^\CD_{t_k}(X_{t_{k}},V_{t_{k}}) - \nabla_{V_{t_k}}\log \rho_{t_k}(X_{t_k},V_{t_k})\right\|^2,
	\end{equation}
	where the weights $\displaystyle  \frac{\Delta\tilde{\Gamma}_{k}}{\tilde{\Gamma}(T)}$  are defined as in Definition~\ref{definition:score estimation error}.
\end{definition}

\begin{assumption}[Lipschitz Score for Critically Damped \ac{ddpm}]
	\label{assumption:Lipschitz score:CD}
	The score function $\nabla_V\log \rho_t(\cdot,\cdot)$ is $L$-Lipschitz continuous on $\mathbb{R}^{2d}$ uniformly across all $t \in [0,T]$.
\end{assumption}

Equipped with these adaptations, we establish the \Renyi divergence error bound for critically damped \acp{ddpm} in the following theorem. Recall that $\displaystyle \Delta_{\max} = \max_{0 \le k \le K-1}\int_{t_k}^{t_{k+1}}\bar{\gamma}(\tau)\d \tau$ denotes the maximum effective discretization step.

\begin{theorem}[\Renyi Divergence Error of Critically Damped \acp{ddpm}]
	\label{theorem:Renyi divergence error of Critically Damped DDPM}
	Suppose Assumptions~\ref{assumption:data distribution} and \ref{assumption:Lipschitz score:CD} hold. Let $\widehat{\mu}$ denote the $X$-marginal distribution generated by the critically damped \ac{ddpm} in~\eqref{eq:practical_reverse_CD}. Let $\rho$ denote the path measure of the exact reverse process $[X,V]^\top$ in~\eqref{eq:exact_reverse_CD}. For any order $\alpha > 1$, the \Renyi divergence of the target distribution $\mu$ relative to the generated distribution $\widehat{\mu}$ satisfies
	\begin{equation}
		\RenyiDivergence{\mu}{\widehat{\mu}} \lesssim \cE^{\CD,M} + \cE^{\CD,S} + \cE^{\CD,D} + o\left(\cE^{\CD,D}\right),
		\label{eq:Renyi divergence error of DDPM:CD}
	\end{equation}
	where $o\left(\cE^{\CD,D}\right)$ denotes a higher-order term satisfying $o\left(\cE^{\CD,D}\right)/\cE^{\CD,D}\to 0$ as $\Delta_{\max}\to 0$.
	The mixing error $\cE^{\CD,M}$, score estimation error $\cE^{\CD,S}$, and discretization error $\cE^{\CD,D}$ are given by
	\begin{align}
		\cE^{\CD,M} &\coloneqq \exponential{-\frac{2\left(\widetilde{\Gamma}(T) - \log(1+\widetilde{\Gamma}(T))\right)}{2\alpha-1}}\RenyiDivergence[2\alpha-1]{\mu}{\xi^d}, \label{eq:mixing error:CD}\\
		\cE^{\CD,S} &\coloneqq \frac{1}{\alpha-1}\log\expectation{\rho}{\exponential{16(2\alpha-1)(4\alpha-1)\widetilde{\Gamma}(T)\scoreErrorCD(X,V)}}, \label{eq:score error:CD}\\ 
		\cE^{\CD,D} &\coloneqq \frac{\alpha^2}{\alpha-1}L^2\widetilde{\Gamma}(T) \left(d \Delta_{\max} + \left(m^2+\sigma^2d\right)\Delta_{\max}^2\right). \label{eq:discretization error:CD}
	\end{align}
\end{theorem}

\begin{remark}[Constant Factor Differences in Score and Discretization Error]
	\label{remark:constant_factor_differences}
	Compared with Theorem~\ref{theorem:Renyi divergence error of DDPM}, both $\cE^{\CD,S}$ and $\cE^{\CD,D}$ differ from their standard counterparts $\cE^{S}$ and $\cE^{D}$ only by constant factors. First, the coefficient in the exponent of $\cE^{\CD,S}$ is twice that of $\cE^{S}$ (16 versus 8). Second, since any constant factor associated with $\cE^{\CD,D}$ is absorbed into the $\lesssim$ notation, $\cE^{\CD,D}$ retains an identical expression to $\cE^{D}$.
\end{remark}
\begin{remark}[Logarithmic Difference in Effective Running Time]
	Critically damped \acp{ddpm} introduce a logarithmic term $ \log(1+\widetilde{\Gamma}(T))$ in the exponent of the mixing error $\cE^{\CD,M}$, which is absent for standard \acp{ddpm}. Notably, this logarithmic penalty can be absorbed into a constant factor $C$ in the exponent, yielding
	\begin{equation}
		\cE^{\CD,M} \le \exponential{-\frac{2C\widetilde{\Gamma}(T)}{2\alpha-1}}\RenyiDivergence[2\alpha-1]{\mu}{\xi^d},
	\end{equation}
	provided $\widetilde{\Gamma}(T) > \Gamma_0$ for some fixed threshold $\Gamma_0 > 0$. This follows because $x - \log(1+x) > Cx$ holds for all $x > \Gamma_0$ whenever $C < 1 - \frac{\log(1+\Gamma_0)}{\Gamma_0}$. 
	In comparison, this exponential decay in mixing error aligns with established results for critically damped \acp{ddpm} in \ac{tv} distance~\cite{chen2022sampling}, \ac{kl} divergence~\cite{conforti2024klconvergenceguaranteesscore}, and Wasserstein distance~\cite{strasman2026wasserstein}.
\end{remark}

In the remainder of this section, we present a proof sketch for Theorem~\ref{theorem:Renyi divergence error of Critically Damped DDPM}, with full technical details deferred to Appendix~\ref{section:appendix:proof of theorem:CD}.

\begin{proof}[Proof Sketch]
	Recall that the standard \ac{ou} process satisfies $X^\OU_t \stackrel{\d}{=} e^{-\Gamma(t)} X^\OU_0 + \sqrt{1 - e^{-2\Gamma(t)}} \xi^d$ (\cf Proposition~\ref{proposition:properties of OU}\ref{enumerate:properties of OU process:2}), whereas the critically damped \ac{ou} process satisfies $X^\CD_t \stackrel{\d}{=} a^X(t) X^\CD_0 + \sqrt{1 - a^X(t)^2} \xi^d$ with $a^X(t) \coloneqq e^{-\Gamma(t)}(1 + \Gamma(t))$ (\cf Proposition~\ref{proposition:properties of CD}\ref{enumerate:properties of CDL:3}). 
	Consequently, the marginal process $X^\CD_t$ can be identified as a standard \ac{ou} process $\d X^{\OU,\text{std}}_t = - X^{\OU,\text{std}}_t \d t + \sqrt{2}\d W_t$ under the effective time $t_{\text{eff}}(t) \coloneqq \Gamma(t) - \log(1 + \Gamma(t))$.
	
	Therefore, the mixing error $\cE^{\CD,M}$ is obtained directly by replacing $\widetilde{\Gamma}(T)$ in $\cE^{M}$ with $t_{\text{eff}}(T)$. The bounds for the score estimation error $\cE^{\CD,S}$ and discretization error $\cE^{\CD,D}$ then follow using arguments analogous to the proof of Theorem~\ref{theorem:Renyi divergence error of DDPM}.
\end{proof}

\section{Discussion}
\label{section:Discussion}

While \Renyi divergence provides a significantly stronger convergence guarantee than \ac{tv} or \ac{kl} divergence, deriving such bounds naturally requires stronger theoretical assumptions. In this section, we compare our assumptions and conclusions with state-of-the-art results in the literature, and highlight promising directions for future research.

\paragraph{Target Distribution}
We assume that the target distribution is sub-Gaussian, a property leveraged to upper-bound the \ac{cgf} of the quadratic form of random vectors. This condition represents a natural middle ground among common distributional assumptions. It is stronger than merely requiring a finite second moment~\cite{chen2022sampling}, yet weaker than assuming bounded support (as bounded random variables are sub-Gaussian) or assuming an \ac{lsi} (since distributions satisfying an \ac{lsi} exhibit sub-Gaussian tails). 
\lwl{Notably, the sub-Gaussian condition offers a distinct advantage over isoperimetric conditions like the Poincaré inequality or \ac{lsi} by focusing exclusively on tail behavior rather than global landscape geometry. For multi-modal distributions, the Poincaré or \ac{lsi} constant $C$ grows exponentially with the separation distance between modes, rendering the nominally exponential mixing decay $e^{-t/C}$ virtually linear in the running time $t$. In contrast, the sub-Gaussian parameter scales only linearly with the separation distance.} Whether \Renyi divergence convergence guarantees can be extended beyond sub-Gaussian target distributions remains an intriguing open question.

\paragraph{Smoothness of the Score Function}
We assume that the score function $\nabla\log\rho_t$ is spatially Lipschitz continuous for all $t$, matching the uniform smoothness assumption employed in \cite{chen2022sampling}. However, this is stronger than requiring Lipschitz continuity only at $t=0$, as adopted in \cite{chen2023improved}. Furthermore, when score smoothness is omitted entirely, sharp convergence guarantees can still be established for \ac{kl} divergence with~\cite{chen2023improved} or without~\cite{conforti2024klconvergenceguaranteesscore}  an early stopping setting. A key open direction is whether \Renyi divergence convergence can similarly be proved under relaxed smoothness conditions.

\paragraph{Score Estimation Error}
Our main result characterizes the score estimation error via the \ac{cgf}, which requires an $L^p$-error control for every $1 \le p < \infty$. It remains an open question whether this \ac{cgf} characterization can be further refined. \lwl{Notably, the work of \cite{yang2022convergence-v2} provided a negative result in the related setting of inexact \ac{lmc}, demonstrating that a purely $L^p$-score error assumption for any fixed $1 \le p < \infty$ is insufficient to guarantee convergence even in \ac{tv} distance.}

Recently, powerful methodologies have been developed to analyze $L^2$-score error, including differential inequality approaches \cite{chen2023improved} and stochastic localization techniques \cite{benton2024nearlydlinearconvergencebounds}, both achieving improved convergence rates. However, because the \ac{cgf} cannot be upper-bounded by moments without additional constraints (\eg, Bernstein-type conditions), these $L^2$-focused advances do not directly translate to the \Renyi divergence framework. Adapting such advanced techniques to sharpen \Renyi divergence error bounds remains an exciting avenue for future work.

\section{Conclusion}
\label{section:Conclusions}

In this paper, we analyzed the convergence of \ac{ddpm}-generated distributions under the \Renyi divergence. After highlighting the benefits of adopting \Renyi divergence over other standard discrepancy measures, we established an upper bound on the \Renyi divergence between general diffusion processes with mismatched drift coefficients via a Girsanov-type change-of-measure argument.
Applying this framework to \acp{ddpm}, we derived explicit error bounds under a sub-Gaussian target distribution and a Lipschitz continuous score function.
\lwl{The error bound decomposes into an exponentially decaying mixing error, a score estimation error characterized via the \ac{cgf}, and a discretization error scaling asymptotically linearly with the effective step size. This establishes the first \Renyi characterization of \acp{ddpm}, under assumptions that avoid the isoperimetric conditions and $L^\infty$-bounded score error required by inexact \ac{lmc}.}
Our analysis also yields an asymptotic iteration complexity of $\widetilde{\bigTheta}\left(\alpha^3 L^2 d / \epsilon\right)$ to achieve an $\epsilon$-accurate distribution in $\alpha$-\Renyi divergence, demonstrating that controlling the stronger \Renyi divergence incurs at most an additional $\alpha^3$ factor.
Furthermore, we extended our results to critically damped \acp{ddpm}, demonstrating they inherit the same upper bounds up to a mild logarithmic factor.
Finally, our findings open several promising avenues for future research, including extending \Renyi divergence guarantees beyond sub-Gaussian distributions and uniform Lipschitz score functions, as well as exploring alternative proof techniques for \Renyi divergence error analysis.

%% file: appendix.tex
\section{Proof of Proposition~\ref{proposition:KL-close distributions can diverge under Renyi divergence}}
\label{section:Appendix:Why Renyi}
	\begin{proof}
	Let $(\Omega,d)$ be a metric space and $P$, $Q$ be probabilities supported on $\Omega$. Denote by $D$ the diameter of $\Omega$ and assume $D<\infty$.
	Let $A\subseteq \Omega$ be a Borel subset and define $P$ and $Q$ as follows.
	First, let $P$ be any Borel probability measure \st $P(A) = \delta$ and $P(\Omega-A) = 1-\delta$.
	Next, define $Q$ via its Radon-Nikodym derivative relative to $P$ as
	\begin{equation}
		\RadonNikodym{Q}{P}(\omega) = 
		\begin{cases}
			e^{-\kappa} & \text{if } \omega \in A, \\
			\frac{1-\delta e^{-\kappa}}{1-\delta} & \text{otherwise.}
		\end{cases}
	\end{equation}
	Then it follows that $Q(A) = \delta e^{-\kappa}$ and $Q(\Omega-A) = 1 - \delta e^{-\kappa}$. In what follows, we prove the proposition by choosing proper $\delta\in(0,1)$ and $\kappa>0$.
	Since the density $\RadonNikodym{Q}{P}$ is less than 1 on $A$ and greater than 1 on $\Omega - A$, we have 
	\begin{equation}
		\totalVariation{P}{Q} = \vert{}P(A) - Q(A)\vert{} = \delta - \delta e^{-\kappa} < \delta.
	\end{equation}
	By the coupling lemma of \ac{tv}, there exists a coupling $\pi$ of $P$ and $Q$ (\ie, a joint distribution $\pi$ on $\Omega \times \Omega$ with marginals $P$ and $Q$) such that $\totalVariation{P}{Q} = \pi(X\neq Y)$. Therefore,
	\begin{equation}
		\WassersteinDistance[p]{P}{Q} \leq \left(\int_{x \neq y} d(x,y)^p \d\pi(x,y)\right)^{1/p} \leq D \left(\int_{x \neq y} \d \pi(x,y)\right)^{1/p} 
		= D \cdot \totalVariation{P}{Q}^{1/p}
		< D\delta^{1/p}.
	\end{equation}
	The KL-divergence and \Renyi divergence follow by direct calculation. 
	\begin{equation}
		\KLDivergence{P}{Q} = \delta \log \frac{\delta}{\delta e^{-\kappa}}
		+ \underbrace{(1-\delta) \log \frac{1-\delta}{1-\delta e^{-\kappa}}}_{<0} < \delta\kappa, 
	\end{equation}
	\begin{equation}
		\RenyiDivergence{P}{Q} = \frac{1}{\alpha-1}\log\Big(\delta^\alpha\left(\delta e^{-\kappa}\right)^{1-\alpha} + \underbrace{(1-\delta)^\alpha\left(1-\delta e^{-\kappa}\right)^{1-\alpha}}_{>0}\Big) >   \frac{1}{\alpha-1}\log\left(\delta e^{(\alpha-1)\kappa} \right) = \kappa + \frac{\log \delta}{\alpha-1}.
		\label{eq:appendix:Renyi divergence between P and Q}
	\end{equation}
	Setting $\kappa =\dfrac{\eta}{2\delta}$ yields $\KLDivergence{P}{Q}<\eta$. 
	Then, it remains to find a solution $\delta$ such that $\delta < \eta$, $\displaystyle D\delta^{1/p} < \eta$, and $\dfrac{\eta}{2\delta} + \dfrac{\log \delta}{\alpha-1} > M$. To tackle the last constraint, we use the inequality $\log x \leq 2(\sqrt{x} - 1) < 2\sqrt{x}$, and it suffices to find $\delta$ to
	satisfy a stricter condition $\dfrac{\eta}{2\delta} - \dfrac{2}{(\alpha-1)\sqrt{\delta}} \ge M$.
	Solving this inequality yields
	\begin{equation}
		\delta<\delta^* \coloneqq \frac{\eta^2}{\left( \frac{2}{\alpha-1} + \sqrt{\frac{4}{(\alpha-1)^2} + 2\eta M} \right)^2}.
	\end{equation}
	Combining with the first two constraints, it suffices to choose $\delta = \dfrac{1}{2} \min\left\{ \eta, \left(\dfrac{\eta}{D}\right)^p, \delta^* \right\}$.
\end{proof}

\section{Proof of Theorem~\ref{theorem:Renyi divergence bound}}
\label{section:Appendix:proof of theorem:Renyi divergence bound}

The proof relies on a change-of-measure argument via Girsanov's theorem. For completeness,  we state the classic result below.

\begin{theorem}[Girsanov Theorem, \cite{le2016brownian}]
	\label{theorem:Girsanov}
	Let $(W_t)_{t\in[0,T]}$ be a $d$-dimensional Brownian motion under a path measure $Q$, and let $(b_s)_{s\in[0,T]}$ be an $\mathbb{R}^d$-valued, $W$-adapted stochastic process satisfying $\expectation{Q}{\int_0^T \|b_s\|^2 \d s} < \infty$.
	Let $\cE(M)_t$ denote the stochastic exponential of the martingale $M_t \coloneqq \int_0^t \langle b_s, \d W_s \rangle$, defined as
	\begin{equation}
		\cE(M)_t \coloneqq \exp\left( \int_0^t \langle b_s, \d W_s \rangle - \frac{1}{2}\int_0^t \|b_s\|^2 \d s \right).
	\end{equation}
	If $\expectation{Q}{\cE(M)_T} = 1$, then under the path measure $P$ defined via the Radon--Nikodym derivative $\RadonNikodym{P}{Q} = \cE(M)_T$, the process
	\begin{equation}
		\widetilde{W}_t \coloneqq W_t - \int_0^t b_s \d s, \quad t \in [0, T],
	\end{equation}
	is a standard $d$-dimensional Brownian motion under $P$.
\end{theorem}

Before completing the main proof of Theorem~\ref{theorem:Renyi divergence bound}, we establish the following auxiliary lemma.

\begin{lemma}[Moments of Stochastic Exponential]
	\label{lemma:moments of exponential martingale}
	Let $M_t$ be a local martingale such that $M_0=0$ almost surely. Then, for any $\alpha,\beta>1$, 
	\begin{align}
		\expectation{}{\cE(M)_t^\alpha}
		&\leq \expectation{}{\exp\left(\frac{\alpha\beta^*(\alpha\beta-1)}{2}[M]_t\right)}^{1 / \beta^*}, 
		\label{eq:moments of exponential martingale:quadratic variation} \\
		\expectation{}{\cE(M)_t^{-\alpha}}
		&\leq \expectation{}{\exp\left(\frac{\alpha\beta^*(\alpha\beta+1)}{2}[M]_t\right)}^{1 / \beta^*}.
		\label{eq:moments of exponential martingale:negative moment}
	\end{align}
\end{lemma}

\begin{proof}[Proof of Lemma~\ref{lemma:moments of exponential martingale}]
	\begin{align}
		\expectation{}{\cE(M)_t^\alpha}
		&= \expectation{}{\exp\left(\alpha M_t - \frac{\alpha}{2}[M]_t\right)} \\
		&= \expectation{}{\exp\left(\alpha M_t - \frac{\alpha^2\beta}{2}[M]_t\right) \exp\left(\frac{\alpha(\alpha\beta-1)}{2} [M]_t\right)} \\
		&\leq \expectation{}{\cE\left(\alpha\beta M\right)_t}^{1/\beta} \expectation{}{\exp\left(\frac{\alpha\beta^*(\alpha\beta-1)}{2} [M]_t\right)}^{1/\beta^*} \\
		&\leq \expectation{}{\exp\left(\frac{\alpha\beta^*(\alpha\beta-1)}{2} [M]_t\right)}^{1/\beta^*}. \nonumber
	\end{align}
	The first inequality follows from Hölder's inequality, and the second inequality follows because $\cE\left(\alpha\beta M\right)_t$ is a supermartingale, and hence $\expectation{}{\cE\left(\alpha\beta M\right)_t}\leq \expectation{}{\cE\left(\alpha\beta M\right)_0} = 1$. This proves~\eqref{eq:moments of exponential martingale:quadratic variation}. Inequality~\eqref{eq:moments of exponential martingale:negative moment} follows by an analogous argument.
	\begin{align}
		\expectation{}{\cE(M)_t^{-\alpha}}
		&= \expectation{}{\exp\left(-\alpha M_t + \frac{\alpha}{2}[M]_t\right)} \\
		&= \expectation{}{\exp\left(-\alpha M_t - \frac{\alpha^2\beta}{2}[M]_t\right) \exp\left(\frac{\alpha(\alpha\beta+1)}{2} [M]_t\right)} \\
		&\leq \expectation{}{\cE\left(-\alpha\beta M\right)_t}^{1/\beta} \expectation{}{\exp\left(\frac{\alpha\beta^*(\alpha\beta+1)}{2} [M]_t\right)}^{1/\beta^*} \\
		&\leq \expectation{}{\exp\left(\frac{\alpha\beta^*(\alpha\beta+1)}{2} [M]_t\right)}^{1/\beta^*}. \nonumber
	\end{align}
\end{proof}

\begin{lemma}[Data Processing Inequality for \Renyi Divergence \cite{van2014renyi}]
	\label{lemma:data processing inequality of Renyi divergence}
	Let $P$ and $Q$ be probability distributions defined on a measurable space $(\Omega, \mathcal{F})$. For any $\alpha\in[0,\infty]$ and any sub-$\sigma$-field $\mathcal{G}\subseteq\mathcal{F}$,
	\begin{equation}
		\label{eq:data processing inequality of Renyi divergence}
		\RenyiDivergence{P|_{\mathcal{G}}}{Q|_{\mathcal{G}}}\leq \RenyiDivergence{P}{Q},
	\end{equation}
	where $P|_{\mathcal{G}}$ denotes the restriction of $P$ to $\mathcal{G}$.
\end{lemma}

\begin{proof}[Proof of Theorem~\ref{theorem:Renyi divergence bound}]
	Define the stochastic process $X^Z_t\coloneqq \frac{1}{\sigma_t}\left(\widehat{b}_t(Z_t) - b_t(Z_t)\right)$,
	where the superscript on $X$ indicates its dependence on $Z$. 
	First, assume that Novikov's condition $\expectation{\rho}{\exponential{\frac{1}{2}\int_0^t \left\|Z^X_s\right\|^2 \d s}}<\infty$ holds. This implies that $\cE\left(\int_0^t \left\langle X^Z_s, \d W_s\right\rangle \right)$ is a uniformly integrable martingale (\cf~\cite[Theorem 5.23]{le2016brownian}). 
	Thus, if $(\Omega, \cF, \{\cF_t\}_{t\geq 0}, P)$ is the filtered probability space on which the process $Z_t$ is defined, we can define another probability measure $Q$ on $(\Omega,\cF)$ via its density relative to $P$, satisfying
	\begin{equation}
		\left.\RadonNikodym{Q}{P}\right|_{\cF_t} = \cE\left(\int_0^t \left\langle X^Z_s,\d W_s\right\rangle \right).
	\end{equation}
	By the Girsanov theorem (\cf Theorem~\ref{theorem:Girsanov}), there exists a $Q$-Brownian motion $\widetilde{W}_t$ whose relation with the $P$-Brownian motion $W_t$  is characterized by the \ac{sde} $\displaystyle \d W_t = X^Z_t \d t + \d \widetilde{W}_t$. Inserting this result into~\eqref{eq:target diffusion} yields
	\begin{equation}
		\d Z_t = \widehat{b}_t(Z_t) \d t + \sigma_t \d \widetilde{W}_t.
		\label{eq:appendix:target diffusion under a change of measure}
	\end{equation}
	At time $0$, $\left.\RadonNikodym{Q}{P}\right|_{\cF_0} = 1$ almost surely under $P$, which implies that changing the measure from $P$ to $Q$ leaves the initial distribution $\rho_0$ of $Z_0$ and $\widehat{Z}_0$ unchanged. 
	Therefore, the \acp{sde} in~\eqref{eq:appendix:target diffusion under a change of measure} and~\eqref{eq:learned diffusion} induce the same diffusion process in $\bbR^d$, and we conclude that $(Z_t)_{\#}Q = (\widehat{Z}_t)_{\#}P$ holds for every $t$.\footnote{For a function $f$ mapping $\mathcal{X}$ to $\mathcal{Y}$ and a probability distribution $P$ on $\mathcal{X}$, $f_{\#}P$ denotes the pushforward measure of $P$, defined by $f_{\#}P(A)\coloneqq P(f^{-1}(A))$ for measurable subsets $A\subseteq \mathcal{Y}$.}
	Denoting by $\nu$ the path measure of $Z$ under $Q$ (\ie, $\nu = Z_{\#}Q$) and recalling that $\rho$ is the path measure of $Z$ under $P$, we have
	\begin{align}
		\RenyiDivergence{\widehat{\rho}_t}{\rho_t}
		&= \RenyiDivergence{\nu_t}{\rho_t}  \leq \RenyiDivergence{\nu}{\rho} 
		\label{eq:proof:Renyi divergence bound:1}
		= \frac{1}{\alpha - 1}\log \expectation{\rho}{\left(\RadonNikodym{\nu}{\rho}\right)^\alpha} \\
		&= \frac{1}{\alpha - 1}\log \expectation{\rho}{\cE\left(\int_0^t \left\langle X^Z_s, \d W_s\right\rangle \right)^\alpha} \\
		&\leq
		\frac{1}{(\alpha - 1)\beta^*}\log
		\expectation{\rho}{\exponential{\frac{\alpha\beta^*(\alpha\beta-1)}{2}\int_0^t\left\|X^Z_s\right\|^2 \d s}} .
		\label{eq:proof:Renyi divergence bound:2}
	\end{align}
	In the above, the inequality in~\eqref{eq:proof:Renyi divergence bound:1} follows from the data processing inequality of \Renyi divergence (\cf~Lemma~\ref{lemma:data processing inequality of Renyi divergence}), and~\eqref{eq:proof:Renyi divergence bound:2} follows from~\eqref{eq:moments of exponential martingale:quadratic variation} in Lemma~\ref{lemma:moments of exponential martingale}. This proves~\eqref{eq:Renyi divergence bound:quadratic variation}. The upper bound~\eqref{eq:Renyi divergence bound between Q and Q-hat} follows from~\eqref{eq:moments of exponential martingale:negative moment} in Lemma~\ref{lemma:moments of exponential martingale} via an analogous argument.
	\begin{align}
		\RenyiDivergence{\rho_t}{\widehat{\rho}_t}
		&\leq \RenyiDivergence{\rho}{\nu} = \frac{1}{\alpha - 1}\log \expectation{\rho}{\left(\RadonNikodym{\rho}{\nu}\right)^{\alpha-1}} \\
		&= \frac{1}{\alpha - 1}\log \expectation{\rho}{\cE\left(\int_0^t \left\langle X^Z_s, \d W_s\right\rangle \right)^{-(\alpha-1)}} \\
		&\leq \frac{1}{(\alpha-1)\beta^*}\log
		\expectation{\rho}{\exponential{\frac{(\alpha-1)\beta^*((\alpha-1)\beta+1)}{2}\int_0^t\left\|X^Z_s\right\|^2 \d s}}.
	\end{align}
	Finally, we show that Novikov's condition can be removed for~\eqref{eq:Renyi divergence bound:quadratic variation}. This is because the coefficient $\dfrac{\alpha\beta^*(\alpha\beta-1)}{2}$, as a function of $\beta$ for $\beta>1$, achieves its minimum at $\beta=1+\sqrt{(\alpha-1)/\alpha}$, and the minimum is $\dfrac{\alpha}{2}\left(\sqrt{\alpha} + \sqrt{\alpha-1}\right)^2$. For $\alpha>1$, this minimum is strictly larger than $1/2$. Therefore, Novikov's condition is satisfied whenever the right-hand side of~\eqref{eq:Renyi divergence bound:quadratic variation} is finite. Otherwise, we obtain the trivial bound $\RenyiDivergence{\widehat{\rho}_t}{\rho_t}\leq \infty$.
\end{proof}

\section{Supplementary Materials for the Proof of Theorem~\ref{theorem:Renyi divergence error of DDPM}}
\label{section:Appendix:proof of main results}

In this section, we first provide several technical lemmas in Section~\ref{subsection:appendix:technical lemmas} that are required to prove Proposition~\ref{proposition:upper-bound of Phi(lambda)}. Subsequently, Section~\ref{subsection:appendix:proof of the proposition on Phi-lambda} presents the complete proof of Proposition~\ref{proposition:upper-bound of Phi(lambda)}.

\subsection{Technical Lemmas}
\label{subsection:appendix:technical lemmas}
\begin{lemma}[Weak Triangle Inequality {\cite[Lemma 2.7]{altschuler2024faster}}]
	\label{lemma:weak triangle inequality}
	Let $\alpha>1$ and $\eta\in(0,1)$. For any probability measures $P$, $Q$, and $R$,
	\begin{equation}
		\RenyiDivergence{P}{Q}\leq \frac{\alpha-\eta}{\alpha-1}\RenyiDivergence[\alpha / \eta]{P}{R}
		+ \RenyiDivergence[(\alpha-\eta) / (1-\eta)]{R}{Q}.
	\end{equation}
	In particular, for $\eta = 1/2$,
	\begin{equation}
		\RenyiDivergence{P}{Q}\leq \frac{2\alpha-1}{2(\alpha-1)}\RenyiDivergence[2\alpha]{P}{R}
		+ \RenyiDivergence[2\alpha - 1]{R}{Q}.
		\label{eq:weak triangle inequality}
	\end{equation}
\end{lemma}

\begin{lemma}[Convergence of \Renyi Divergence along Langevin Dynamics under \ac{lsi}]
	\label{lemma:convergence of Renyi divergence along LD under LSI}
	Let $V\colon \bbR^d\to \bbR$ be the potential function, and let $(X_t)_{t\geq 0}$ evolve according to the (time-rescaled) Langevin dynamics $\d X_t = -\gamma(t)\nabla V(X_t)\d t + \sqrt{2\gamma(t)}\d W_t$. 
	Denote by $\rho_t$ the distribution of $X_t$, and by $\pi\propto\exp(-V)$ the stationary distribution of $X_t$. If $\pi$ satisfies \ac{lsi} with parameter $\eta$, \ie, for any smooth function $f$,
	\begin{equation}
		\expectation{\pi}{f^2\log f^2} - \expectation{\pi}{f^2}\log \expectation{\pi}{f^2} = \ent{\pi}{f^2} \leq \frac{2}{\eta}\expectation{\pi}{\|\nabla f\|^2},
	\end{equation}
	then for any $\alpha>1$,
	\begin{equation}
		\RenyiDivergence{\rho_t}{\pi}\leq \exponential{-\frac{2\eta}{\alpha}\int_0^t\gamma_s \d s}\RenyiDivergence{\rho_0}{\pi}.
		\label{eq:lemma:convergence of Renyi divergence along Langevin dynamics under LSI}
	\end{equation}
\end{lemma}
For the \ac{ou} process, the potential function $V(x) = \frac{1}{2}\|x\|^2$ is $1$-strongly convex, and by the Bakry-Émery criterion, the \ac{lsi} holds with $\eta=1$.
The case $\gamma(t) \equiv 1$ was proved in~\cite[Theorem 2]{vempala2019rapid}, and their proof can be directly extended to general $\gamma(t)$. We provide the proof for completeness. 

\begin{proof}[Proof of Lemma~\ref{lemma:convergence of Renyi divergence along LD under LSI}]
	We first show that
	\begin{equation}
		\partial_t \RenyiDivergence{\rho_t}{\pi} = -\alpha \gamma(t) \frac{\expectation{\pi}{\left(\frac{\rho_t}{\pi}\right)^\alpha\left\|\nabla\log\frac{\rho_t}{\pi}\right\|^2}}{\expectation{\pi}{\left(\frac{\rho_t}{\pi}\right)^\alpha}}.
		\label{eq:appendix:partial t of Renyi divergence along Langevin dynamics}
	\end{equation}
	Assuming~\eqref{eq:appendix:partial t of Renyi divergence along Langevin dynamics} holds, we invoke Lemma 5 of~\cite{vempala2019rapid}, which reads
	\begin{equation}
		\frac{\expectation{\pi}{\left(\frac{\rho_t}{\pi}\right)^\alpha\left\|\nabla\log\frac{\rho_t}{\pi}\right\|^2}}{\expectation{\pi}{\left(\frac{\rho_t}{\pi}\right)^\alpha}} \geq \frac{2\eta}{\alpha^2} \RenyiDivergence{\rho_t}{\pi}.
	\end{equation}
	Inserting this result into~\eqref{eq:appendix:partial t of Renyi divergence along Langevin dynamics} yields
	\begin{equation}
		\partial_t \log \RenyiDivergence{\rho_t}{\pi} \leq -\frac{2\eta}{\alpha}\gamma(t).
	\end{equation}
	Integrating over $[0,t]$ yields the desired result~\eqref{eq:lemma:convergence of Renyi divergence along Langevin dynamics under LSI}. Thus, it remains to prove~\eqref{eq:appendix:partial t of Renyi divergence along Langevin dynamics}. This follows from the Fokker-Planck equation for $X_t$ and integration by parts.
	{
		\allowdisplaybreaks
		\begin{align}
			\partial_t \RenyiDivergence{\rho_t}{\pi} 
			&= \frac{1}{\alpha-1} \partial_t \log \expectation{\pi}{\left(\frac{\rho_t}{\pi}\right)^\alpha} \\
			&= \frac{\alpha}{\alpha-1} \frac{1}{\expectation{\pi}{\left(\frac{\rho_t}{\pi}\right)^\alpha}} \int \left(\frac{\rho_t}{\pi}\right)^{\alpha-1} \partial_t \rho_t \, \d x \\
			&= \frac{\alpha}{\alpha-1} \frac{1}{\expectation{\pi}{\left(\frac{\rho_t}{\pi}\right)^\alpha}} \int  \left(\frac{\rho_t}{\pi}\right)^{\alpha-1} \gamma(t) \nabla \cdot \left(\rho_t \nabla \log \frac{\rho_t}{\pi}\right) \d x \\
			&= -\frac{\alpha\gamma(t)}{\alpha-1} \frac{1}{\expectation{\pi}{\left(\frac{\rho_t}{\pi}\right)^\alpha}} \int \rho_t \left\langle \nabla \left(\frac{\rho_t}{\pi}\right)^{\alpha-1}, \nabla \log \frac{\rho_t}{\pi} \right\rangle \d x \\
			&= -\alpha\gamma(t) \frac{1}{\expectation{\pi}{\left(\frac{\rho_t}{\pi}\right)^\alpha}} \int \rho_t \left(\frac{\rho_t}{\pi}\right)^{\alpha-1} \left\langle \nabla \log \frac{\rho_t}{\pi}, \nabla \log \frac{\rho_t}{\pi} \right\rangle \d x \\
			&= -\alpha\gamma(t) \frac{\expectation{\pi}{\left(\frac{\rho_t}{\pi}\right)^\alpha\left\|\nabla\log\frac{\rho_t}{\pi}\right\|^2}}{\expectation{\pi}{\left(\frac{\rho_t}{\pi}\right)^\alpha}}.
		\end{align}
	}
\end{proof}

\begin{lemma}[Score Perturbation for DDPM {\cite{lee2022convergence, chen2022sampling}}]
	\label{lemma:score perturbation}
	Let $X\in\bbR^d$ be a random variable with density $\pi^X(x) \propto e^{-V(x)}$. Let $Y = AX+\Sigma^{1/2}\,\Xi^d$, where $A\in\bbR^{d\times d}$ is a fixed matrix , $\Sigma\in\bbR^{d\times d}$ is a positive definite matrix, and $\Xi^d$ is a standard $d$-dimensional Gaussian random vector. Assume that $\|A-I_d\|_2\leq \zeta$ and $\nabla V$ is $L$-Lipschitz with $L\leq \dfrac{1}{4\|\Sigma\|_2}$. Then,
	\begin{equation}
		\left\|\nabla\log\frac{\pi^Y}{\pi^X}(x)\right\|\lesssim L\sqrt{\|\Sigma\|_2d} + L\zeta\|x\| + \left(\zeta + L\|\Sigma\|_2\right)\|\nabla V(x)\|,
	\end{equation}
	where $\pi^Y$ denotes the distribution of $Y$.
\end{lemma}

For the OU process~\eqref{eq:OU process} and $0<s<t$, we have
\begin{equation}
	X^\OU_t\stackrel{\d}{=}e^{ -\left(\Gamma(t) - \Gamma(s)\right)} X^\OU_s + \sqrt{1-e^{-2\left(\Gamma(t) - \Gamma(s)\right)}}\,\Xi^d.
\end{equation}
Thus, by choosing $A = e^{-\left(\Gamma(t) - \Gamma(s)\right)} I_d$ and $\Sigma = \left(1-e^{-2\left(\Gamma(t) - \Gamma(s)\right)}\right)I_d$, we obtain $\zeta = \cO\left(\Gamma(t)-\Gamma(s)\right)$ and $\|\Sigma\|_2 = \cO\left(\Gamma(t)-\Gamma(s)\right)$. For $t\in[t_k,t_{k+1})$, the above lemma yields
\begin{align}
	\left\|\nabla\log \frac{\rho_{t_k}}{\rho_{t}}\left(X_{t_k}\right)\right\|^2 &\lesssim L^2\|\Sigma\|_2d + L^2\zeta^2\left\|X_{t_k}\right\|^2 + \left(\zeta + L\|\Sigma\|_2\right)^2\left\|\nabla \log \rho_t\left(X_{t_k}\right)\right\|^2\\
	&\lesssim L^2\left(\bar{\Gamma}(t_k)-\bar{\Gamma}(t)\right)d + L^2\left(\bar{\Gamma}(t_k)-\bar{\Gamma}(t)\right)^2\left\|X_{t_k}\right\|^2 + L^2\left(\bar{\Gamma}(t_k)-\bar{\Gamma}(t)\right)^2\left\|\nabla \log \rho_t\left(X_{t_k}\right)\right\|^2\\
	&= L^2\left(\tilde{\Gamma}(t)-\tilde{\Gamma}(t_k)\right)d + L^2\left(\tilde{\Gamma}(t)-\tilde{\Gamma}(t_k)\right)^2\left\|X_{t_k}\right\|^2 + L^2\left(\tilde{\Gamma}(t)-\tilde{\Gamma}(t_k)\right)^2\left\|\nabla \log \rho_t\left(X_{t_k}\right)\right\|^2.
\end{align}

\begin{lemma}[Score Concentration~\cite{altschuler2023shifted}]
	\label{lemma:score concentration}
	Let $\pi\propto\exp(-V)$ be a probability measure with potential $V$. If the Hessian of $V$ satisfies $-L I \preceq \nabla^2 V \preceq L I$, then $\nabla V$ is $\sqrt{L}$-sub-Gaussian, in the sense 
	\begin{equation}
		\expectation{\pi}{\exp\left(\lambda\langle v,\nabla V\rangle\right)} \leq \exp\left(\frac{\lambda^2L}{2}\right),
	\end{equation}
	for any $v\in\bbR^d$ satisfying $\|v\|=1$.
\end{lemma}

\begin{lemma}[Moment of Quadratic Sub-Gaussian Variables~\cite{hsu2012tail}]
	\label{lemma:moments of quadratic sub-Gaussian}
	Let $X\in\bbR^d$ be a sub-Gaussian random vector with parameter $\sigma$, \ie, for any $v\in\bbR^d$,
	\begin{equation}
		\expectation{}{\exp\left(\langle v, X-\expectation{}{X}\rangle\right)}\leq \exp\left(\frac{\|v\|^2\sigma^2}{2}\right).
	\end{equation}
	Then, for $\lambda<\dfrac{1}{2\sigma^2}$,
	\begin{equation}
		\expectation{}{\exp\left(\lambda\|X\|^2\right)}\leq \exp\left(\sigma^2d\lambda + \frac{\|\expectation{}{X}\|^2\lambda + \sigma^4\lambda^2d}{1-2\sigma^2\lambda}\right).
	\end{equation}
\end{lemma}

\subsection{Proof of Proposition~\ref{proposition:upper-bound of Phi(lambda)}}
\label{subsection:appendix:proof of the proposition on Phi-lambda}

Recall that $\Phi(\lambda)$ is defined as $\displaystyle \Phi(\lambda) = 
	\log\expectation{X\sim\rho}{\exponential{\lambda \sum_{k=0}^{K-1} \int_{t_k}^{t_{k+1}}\bar{\gamma}(t)\left\|S_{t_k}\left(X_{t_k}\right) - \nabla\log \rho_{t}(X_t)\right\|^2 \d t}}$. For $t\in[t_k,t_{k+1})$, the integrand can be decomposed as
\begin{multline}
	\bar{\gamma}(t)\left\|S_{t_k}\left(X_{t_k}\right) - \nabla\log \rho_{t}(X_t)\right\|^2
	\leq 3\bar{\gamma}(t)\Big(\left\|S_{t_k}\left(X_{t_k}\right) - \nabla\log \rho_{t_k}\left(X_{t_k}\right)\right\|^2 \\
	+ \left\|\nabla\log \rho_{t_k}\left(X_{t_k}\right) - \nabla\log \rho_{t}\left(X_{t_k}\right)\right\|^2 + \left\|\nabla\log \rho_{t}\left(X_{t_k}\right) - \nabla\log \rho_{t}\left(X_{t}\right)\right\|^2 \Big).
\end{multline}
Applying the score perturbation lemma (\cf\ Lemma~\ref{lemma:score perturbation}) to the second term and using the $L$-Lipschitz property of $\nabla\log \rho_t$ (\cf\ Assumption~\ref{assumption:Lipschitz score}) for the third term, we have, for some constant $C>0$,
\begin{multline}
	\bar{\gamma}(t)\left\|S_{t_k}\left(X_{t_k}\right) - \nabla\log \rho_{t}(X_t)\right\|^2 
	\leq 3\bar{\gamma}(t)
	\left\|S_{t_k}\left(X_{t_k}\right) - \nabla\log \rho_{t_k}\left(X_{t_k}\right)\right\|^2 \\
	 + CL^2\bar{\gamma}(t)\left(
	d\Delta\tilde{\Gamma}_{t_k,t} + \left\|\nabla \log \rho_t\left(X_{t}\right)\right\|^2 \left(\Delta\tilde{\Gamma}_{t_k,t}\right)^2 
	+ \left\|X_{t_k}\right\|^2 \left(\Delta\tilde{\Gamma}_{t_k,t}\right)^2 + \left\|X_t-X_{t_k}\right\|^2\right).
	\label{eq:upperbound of the integrand of the path integral}
\end{multline}
Here, we use $\Delta\tilde{\Gamma}_{t_k,t}$ to denote $\tilde{\Gamma}(t) - \tilde{\Gamma}(t_k)$.
Inserting~\eqref{eq:upperbound of the integrand of the path integral} into~\eqref{eq:an expression that we want to upper-bound} yields
\begin{equation}
	\Phi(\lambda)
	\leq \log\expectation{\rho}{\prod_{i=1}^{5}\cI_i(\lambda)},
	\label{eq:upper bound of the form prod I_i}
\end{equation}
where we define
{\allowdisplaybreaks
\begin{align}
	\cI_1(\lambda) &\coloneqq \exponential{3\lambda\sum_{k=0}^{K-1}\int_{t_{k}}^{t_{k+1}}\bar{\gamma}(t)\left\|S_{t_k}\left(X_{t_k}\right)-\nabla\log\rho_{t_k}\left(X_{t_k}\right)\right\|^2\d t},\\
	\cI_2(\lambda) &\coloneqq
	\exponential{\lambda CL^2\sum_{k=0}^{K-1}\int_{t_{k}}^{t_{k+1}}\bar{\gamma}(t) \Delta\tilde{\Gamma}_{t_k,t}\d t},\\
	\cI_3(\lambda) &\coloneqq
	\exponential{\lambda CL^2\sum_{k=0}^{K-1}\int_{t_{k}}^{t_{k+1}}\bar{\gamma}(t)\left\|\nabla\log\rho_{t}\left(X_{t}\right)\right\|^2\left(\Delta\tilde{\Gamma}_{t_k,t}\right)^2\d t},\label{eq:definition of I_3}\\
	\cI_4(\lambda) &\coloneqq
	\exponential{\lambda CL^2\sum_{k=0}^{K-1}\int_{t_{k}}^{t_{k+1}}\bar{\gamma}(t)\left\|X_{t_k}\right\|^2\left(\Delta\tilde{\Gamma}_{t_k,t}\right)^2\d t},
	\label{eq:definition of I_4}\\
	\cI_5(\lambda) &\coloneqq
	\exponential{\lambda CL^2\sum_{k=0}^{K-1}\int_{t_{k}}^{t_{k+1}}\bar{\gamma}(t)\left\|X_t - X_{t_k}\right\|^2\d t}.
	\label{eq:definition of I_5} 
\end{align}
}Since
$\partial_t \Delta\tilde{\Gamma}_{t_k,t} = \bar{\gamma}(t)$, we have
\begin{equation}
	\cI_2(\lambda) = \exponential{\frac{1}{2}\lambda CL^2\sum_{k=0}^{K-1}\left( \tilde{\Gamma}(t_{k+1}) - \tilde{\Gamma}(t_k) \right)^2} = \exponential{\frac{1}{2}\lambda CL^2\sum_{k=0}^{K-1}\left(\Delta \tilde{\Gamma}_k\right)^2}.
	\label{eq:I_2}
\end{equation}
Inserting~\eqref{eq:I_2} into~\eqref{eq:upper bound of the form prod I_i} yields
{\allowdisplaybreaks
\begin{align}
	\log\expectation{\rho}{\prod_{i=1}^{5}\cI_i(\lambda)}
	&= \frac{1}{2}\lambda CL^2\sum_{k=0}^{K-1}\left(\Delta \tilde{\Gamma}_k\right)^2 + \log\expectation{\rho}{\cI_1(\lambda)\cI_3(\lambda)\cI_4(\lambda)\cI_5(\lambda)} \\
	\begin{split}
		&\leq \frac{1}{2}\lambda CL^2  \sum_{k=0}^{K-1}\Delta \tilde{\Gamma}_k \cdot \Delta_{\max}
		+ \frac{3}{4}\log\expectation{\rho}{\cI_1\left(\frac{4}{3}\lambda\right)}
		+ \frac{1}{8}\log\expectation{\rho}{\cI_3\left(8\lambda\right)} \\
		&\qquad + \frac{1}{16}\log\expectation{\rho}{\cI_4\left(16\lambda\right)}
		+ \frac{1}{16}\log\expectation{\rho}{\cI_5\left(16\lambda\right)}
	\end{split} 
	\label{eq:Holder's inequality is used here}\\
	\begin{split}
		&= \frac{1}{2}\lambda CL^2\tilde{\Gamma}(T) \Delta_{\max}+ 
		\frac{3}{4}\log\expectation{\rho}{\exponential{4\lambda\tilde{\Gamma}(T)\scoreError(X)}}
		+ \frac{1}{8}\log\expectation{\rho}{\cI_3\left(8\lambda\right)} \\
		&\qquad + \frac{1}{16}\log\expectation{\rho}{\cI_4\left(16\lambda\right)}
		+ \frac{1}{16}\log\expectation{\rho}{\cI_5\left(16\lambda\right)}.
	\end{split}
	\label{eq:I_1 is linked to the score error}
\end{align}
}In the above, \eqref{eq:Holder's inequality is used here} follows from the generalized \Holder inequality $\int \prod_i f_i \leq \prod_i \left( \int f_i^{\alpha_i} \right)^{1/\alpha_i}$
for non-negative functions $f_i$ and exponents $\alpha_i$ satisfying $\sum_i \frac{1}{\alpha_i} = 1$. Here, the exponents are chosen according to the identity $\frac{3}{4} + \frac{1}{4}\left(\frac{1}{2} + \frac{1}{4} + \frac{1}{4}\right) = 1$. Additionally, \eqref{eq:I_1 is linked to the score error} follows from Definition~\ref{definition:score estimation error}. Consequently, it remains to upper-bound $\expectation{\rho}{\cI_i(\lambda)}$ for $i \in \{3, 4, 5\}$ and some $\lambda > 0$.

\paragraph{Bound $\expectation{\rho}{\cI_3(\lambda)}$} Following~\eqref{eq:definition of I_3}, we have
\begin{align}
	\expectation{\rho}{\cI_3(\lambda)}
	&= \expectation{\rho}{\exponential{\lambda CL^2\int_{0}^{T}\bar{\gamma}(t)\sum_{k=0}^{K-1}\left\|\nabla\log\rho_t(X_t)\right\|^2\left(\Delta\tilde{\Gamma}_{t_k,t}\right)^2\1{[t_k,t_{k+1})}(t)\d t}} 
	\label{eq:upper-bound I_3:We start here}\\
	&\leq \frac{1}{\tilde{\Gamma}(T)}\int_{0}^{T}\bar{\gamma}(t)\expectation{\rho}{\exponential{
		\lambda CL^2\tilde{\Gamma}(T)\sum_{k=0}^{K-1}\left\|\nabla\log\rho_t(X_t)\right\|^2\left(\Delta\tilde{\Gamma}_{t_k,t}\right)^2
		\1{[t_k,t_{k+1})]}(t)}}\d t 
	\label{eq:upper-bound I_3:Jensen's inequality is used here}\\
	&\leq \sup_{t\in[0,T]} \expectation{\rho}{\exponential{\lambda CL^2\tilde{\Gamma}(T)\sum_{k=0}^{K-1}\left\|\nabla\log\rho_t(X_t)\right\|^2\left(\Delta\tilde{\Gamma}_{t_k,t}\right)^2\1{[t_k,t_{k+1})]}(t)}}
	\label{eq:upper-bound I_3:average is less than supreme}\\
	&\leq \sup_{t\in[0,T]} \expectation{\rho}{\exponential{\lambda CL^2\tilde{\Gamma}(T) \Delta_{\max}^2 \left\|\nabla\log\rho_t(X_t)\right\|^2}}.
	\label{eq:upper-bound I_3:monotonicity of Gamma-tilde is used here}
\end{align}
In the above, \eqref{eq:upper-bound I_3:Jensen's inequality is used here} follows from Jensen's inequality, and~\eqref{eq:upper-bound I_3:monotonicity of Gamma-tilde is used here} follows from the fact that $\Delta\tilde{\Gamma}_{t_k,t}\leq \Delta\tilde{\Gamma}_k$ for $t\in[t_k,t_{k+1})$.
By Lemma~\ref{lemma:score concentration}, $\nabla\log\rho_t$ is $\sqrt{L}$-sub-Gaussian for every $t\in[0,T]$. Hence, we can invoke Lemma~\ref{lemma:moments of quadratic sub-Gaussian} to upper-bound~\eqref{eq:upper-bound I_3:monotonicity of Gamma-tilde is used here}. Notice that
\begin{equation}
	\expectation{\rho_t}{\nabla\log \rho_t} = \int\nabla\rho_t = \nabla \int \rho_t = 0.
\end{equation}
Define $\tilde{\lambda} \coloneqq \lambda CL^2\tilde{\Gamma}(T) \Delta_{\max}^2$, Lemma~\ref{lemma:moments of quadratic sub-Gaussian} then yields
\begin{align}
	\eqref{eq:upper-bound I_3:monotonicity of Gamma-tilde is used here}
	&\leq \exp\left(Ld\tilde{\lambda}+ \frac{L^2\tilde{\lambda}^2d}{1-2L\tilde{\lambda}}\right) 
	= \exp\left(Ld\tilde{\lambda}+ L^2\tilde{\lambda}^2d \left(1+2L\tilde{\lambda}+o\left( L\tilde{\lambda}\right)\right)\right) 
	= \exponential{Ld\tilde{\lambda}+ o\left(Ld\tilde{\lambda}\right)} 
	\label{eq:upper-bound I_3:final result:-1}\\
	&= 1 + Ld\tilde{\lambda}+ o\left(Ld\tilde{\lambda}\right).
	\label{eq:upper-bound I_3:final result}
\end{align}
Here, we assume $L\tilde{\lambda}= o(1)$ and $Ld\tilde{\lambda}= o(1)$ to obtain \eqref{eq:upper-bound I_3:final result:-1}  and \eqref{eq:upper-bound I_3:final result}, respectively. Note that the former condition is implied by the latter.

\paragraph{Bound $\expectation{\rho}{\cI_4(\lambda)}$} 
Similar to the steps~\eqref{eq:upper-bound I_3:We start here}--\eqref{eq:upper-bound I_3:monotonicity of Gamma-tilde is used here}, we have
\begin{align}
	\expectation{\rho}{\cI_4(\lambda)}
	&= \expectation{\rho}{\exponential{\lambda CL^2\int_{0}^{T}\bar{\gamma}(t)\sum_{k=0}^{K-1}\left\|X_{t_k}\right\|^2\left(\Delta\tilde{\Gamma}_{t_k,t}\right)^2\1{[t_k,t_{k+1})}(t)\d t}} 
	\label{eq:upper-bound I_4:move the sum inside the integral} \\
	&\leq \frac{1}{\tilde{\Gamma}(T)}\int_{0}^{T}\bar{\gamma}(t)\expectation{\rho}{\exponential{
			\lambda CL^2\tilde{\Gamma}(T)\sum_{k=0}^{K-1}\left\|X_{t_k}\right\|^2\left(\Delta\tilde{\Gamma}_{t_k,t}\right)^2
			\1{[t_k,t_{k+1})]}(t)}}\d t 
	\label{eq:upper-bound I_4:Jensen's inequality is used here} \\
	&\leq \sup_{t\in[0,T]} \expectation{\rho}{\exponential{\lambda CL^2\tilde{\Gamma}(T)\sum_{k=0}^{K-1}\left\|X_{t_k}\right\|^2\left(\Delta\tilde{\Gamma}_{t_k,t}\right)^2\1{[t_k,t_{k+1})]}(t)}}
	\label{eq:upper-bound I_4:average is less than supreme} \\
	&= \sup_{k\in\{0,1,\dots,K-1\}}\expectation{\rho}{\exponential{\tilde{\lambda} \left\|X_{t_k}\right\|^2}},
	\label{eq:upper-bound I_4:monotonicity of Gamma-tilde is used here}	
\end{align}
where $\tilde{\lambda} = \lambda CL^2\tilde{\Gamma}(T) \Delta_{\max}^2$ has been defined when bounding $\expectation{\rho}{\cI_3(\lambda)}$. Using the property of the \ac{ou} process (\cf Proposition~\ref{proposition:properties of OU}\ref{enumerate:properties of OU process:2})
\begin{equation}
	X^\OU_t \stackrel{\d}{=} e^{-\Gamma(t)}X^\OU_0 + \sqrt{1-e^{-2\Gamma(t)}}\, \Xi^d, 
	\label{eq:distribution of X-OU-t}
\end{equation}
where $\Gamma(t) =  \int_{0}^{t} \gamma(\tau)\d \tau$ and $\Xi^d\sim\xi^d$, we have
\begin{equation}
	X_{t_k} \stackrel{\d}{=} X^\OU_{T-t_k} \stackrel{\d}{=} e^{-\bar{\Gamma}(t_k)}X^\OU_0 + \sqrt{1-e^{-2\bar{\Gamma}(t_k)}}\,\Xi^d.
\end{equation}
Following~\eqref{eq:upper-bound I_4:monotonicity of Gamma-tilde is used here}, we obtain
{\allowdisplaybreaks
\begin{align}
	\expectation{\rho}{\cI_4(\lambda)}
	&\leq \sup_{k\in\{0,1,\dots,K-1\}} \expectation{\rho}{\exponential{\tilde{\lambda}\left\|X_{t_k}\right\|^2}} \\
	&= \sup_{k\in\{0,1,\dots,K-1\}} \expectation{\mu,\xi^d}{ \exponential{\tilde{\lambda}\left\|e^{-\bar{\Gamma}(t_k)}X^\OU_0 + \sqrt{1-e^{-2\bar{\Gamma}(t_k)}}\,\Xi^d\right\|^2}} \\
	&\leq \sup_{k\in\{0,1,\dots,K-1\}} \expectation{\mu,\xi^d}{\exponential{
		2\tilde{\lambda}e^{-2\bar{\Gamma}(t_k)}\left\|X^\OU_0\right\|^2 + 2 \tilde{\lambda} \left(1-e^{-2\bar{\Gamma}(t_k)}\right)\left\|\Xi^d\right\|^2}} \\
	&= \sup_{k\in\{0,1,\dots,K-1\}}  \expectation{\mu}{\exponential{2\tilde{\lambda}e^{-2\bar{\Gamma}(t_k)}\left\|X^\OU_0\right\|^2}}
		\expectation{\xi^d}{ \exponential{2\tilde{\lambda}	\left(1-e^{-2\bar{\Gamma}(t_k)}\right)\left\|\Xi^d\right\|^2}} \\
	&< \expectation{\mu}{\exponential{2\tilde{\lambda}\left\|X^\OU_0\right\|^2}}
	\expectation{\xi^d}{ \exponential{2\tilde{\lambda}\left\|\Xi^d\right\|^2}}.
	\label{eq:upper-bound I_4:product of two expectations}
\end{align}
}By Assumption~\ref{assumption:data distribution}, $\mu$ is $\sigma$-sub-Gaussian. Hence, we can invoke Lemma~\ref{lemma:moments of quadratic sub-Gaussian} to upper-bound the first expectation in~\eqref{eq:upper-bound I_4:product of two expectations}.
\begin{align}
	\expectation{\mu}{\exponential{2\tilde{\lambda} \left\|X^\OU_0\right\|^2}} 
	&\leq \exp\left(2\sigma^2d\tilde{\lambda}  + \frac{2m^2\tilde{\lambda}  + 4\sigma^4\tilde{\lambda}^2d}{1-4\sigma^2\tilde{\lambda}}\right) 
	\label{eq:upper-bound I_4:upper bound of the first expectation:term 1} \\
	&= \exp\left(2\sigma^2d\tilde{\lambda} + \left(2m^2\tilde{\lambda} + 4\sigma^4\tilde{\lambda}^2d\right)\left(1+4\sigma^2\tilde{\lambda} + o\left(\sigma^2\tilde{\lambda}\right)\right)\right) 
	\label{eq:upper-bound I_4:upper bound of the first expectation:-1}\\
	&= \exponential{2\left(m^2+\sigma^2d\right)\tilde{\lambda} + o\left(\left(m^2+\sigma^2d\right)\tilde{\lambda}\right)} \\
	&= 1 + 2\left(m^2+\sigma^2d\right)\tilde{\lambda} + o\left(\left(m^2+\sigma^2d\right)\tilde{\lambda}\right).
	\label{eq:upper-bound I_4:upper bound of the first expectation}
\end{align}
We assume $\sigma^2\tilde{\lambda} = o(1)$ and $\left(m^2+\sigma^2d\right)\tilde{\lambda} = o(1)$ to obtain~\eqref{eq:upper-bound I_4:upper bound of the first expectation:-1} and \eqref{eq:upper-bound I_4:upper bound of the first expectation}, respectively. Note that the former condition is implied by the latter.

Since $\left\|\Xi^d\right\|^2$ is a $\chi^2$-random variable with $d$ degrees of freedom, its moment generating function is given by
\begin{equation}
	\expectation{\xi^d}{ \exponential{2\tilde{\lambda}\left\|\Xi^d\right\|^2}} 
	= \left(1- 4\tilde{\lambda} \right)^{-d/2} 
	= 1 + 2d\tilde{\lambda} + o\left(d\tilde{\lambda}\right). 
	\label{eq:upper-bound I_4:upper bound of the second expectation}
\end{equation}
Here, we assume $d\tilde{\lambda}= o(1)$ to obtain the second equality. Notice that this condition is also implied by $Ld\tilde{\lambda} = o(1)$, which is used in~\eqref{eq:upper-bound I_3:final result}.
Inserting~\eqref{eq:upper-bound I_4:upper bound of the first expectation} and \eqref{eq:upper-bound I_4:upper bound of the second expectation} into \eqref{eq:upper-bound I_4:product of two expectations} yields
\begin{equation}
	\expectation{\rho}{\cI_4(\lambda)} \leq 
	 1+ 2\left(m^2+(1+\sigma^2)d\right)\tilde{\lambda} + o\left(\left(m^2+(1+\sigma^2\right)d)\tilde{\lambda}\right).
	\label{eq:upper-bound I_4:final result}
\end{equation}

\paragraph{Bound $\expectation{\rho}{\cI_5(\lambda)}$} 
Bounding $\expectation{\rho}{\cI_5(\lambda)}$ follows the same technical approach as that used for $\expectation{\rho}{\cI_4(\lambda)}$. For completeness, we present the full derivation below, focusing primarily on where the arguments differ.

\begin{align}
	\expectation{\rho}{\cI_5(\lambda)}
	&= \expectation{\rho}{\exponential{\lambda CL^2\int_{0}^{T}\bar{\gamma}(t)\sum_{k=0}^{K-1}\left\|X_{t} - X_{t_k}\right\|^2\1{[t_k,t_{k+1})}(t)\d t}} 
	\label{eq:upper-bound I_5:move the sum inside the integral} \\
	&\leq \frac{1}{\tilde{\Gamma}(T)}\int_{0}^{T}\bar{\gamma}(t)\expectation{\rho}{\exponential{\lambda CL^2\tilde{\Gamma}(T)\sum_{k=0}^{K-1}\left\|X_{t} - X_{t_k}\right\|^2\1{[t_k,t_{k+1})]}(t)}}\d t 
	\label{eq:upper-bound I_5:Jensen's inequality is used here}\\
	&\leq \sup_{t\in[0,T]} \expectation{\rho}{\exponential{\lambda CL^2\tilde{\Gamma}(T)\sum_{k=0}^{K-1}\left\|X_{t} - X_{t_k}\right\|^2\1{[t_k,t_{k+1})]}(t)}}
	\label{eq:upper-bound I_5:average is less than supreme}\\
	&= \sup_{k\in\{0,1,\dots,K-1\}, t\in[t_k,t_{k+1}]} \expectation{\rho}{\exponential{\lambda CL^2\tilde{\Gamma}(T)\left\|X_{t} - X_{t_k}\right\|^2}}.
	\label{eq:upper-bound I_5:monotonicity is used here}
\end{align}
According to the property of the \ac{ou} process (\cf Proposition~\ref{proposition:properties of OU}\ref{enumerate:properties of OU process:3}), we have, for $0<s<t$,
\begin{equation}
	X^\OU_t - X^\OU_s \stackrel{\d}{=} \left(e^{-\Gamma(t)} - e^{-\Gamma(s)}\right)X^\OU_0 + \sigma_{s,t}\Xi^d,
	\label{eq:distribution of X-OU-t minus X-OU-s}
\end{equation}
where $\displaystyle\sigma_{s,t}^2 = 2\left(1-e^{-(\Gamma(t)-\Gamma(s))}\right) - \left( e^{-\Gamma(t)} - e^{-\Gamma(s)}\right)^2$. Consequently,
\begin{equation}
	X_{t} - X_{t_k} \stackrel{\d}{=} X^\OU_{T-t} - X^\OU_{T-t_k} 
	\stackrel{\d}{=} \left(e^{-\bar{\Gamma}(t)} - e^{-\bar{\Gamma}(t_{k})}\right)X^\OU_0 + \bar{\sigma}_{t_k,t}\Xi^d,
	\label{eq:distribution of X-t-k+1 minus X-t-k}
\end{equation}
with $\bar{\sigma}_{t_k,t}$ given by
\begin{equation}
	\bar{\sigma}_{t_k,t}^2 = \sigma_{T-t_k,T-t}^2 =
	2\left(1-e^{-\left(\bar{\Gamma}(t_{k})-\bar{\Gamma}(t)\right)}\right) - \left( e^{-\bar{\Gamma}(t_{k})} - e^{-\bar{\Gamma}(t)}\right)^2 
	< 2\left(1-e^{-\Delta\tilde{\Gamma}_{t_k,t}}\right) \leq 2\Delta\tilde{\Gamma}_{t_k,t}. \label{eq:upper-bound of tilde-sigma-k} 
\end{equation}
The first inequality follows from the identity $\bar{\Gamma}(t_{k})-\bar{\Gamma}(t) = \Delta\tilde{\Gamma}_{t_k,t}$. Inserting~\eqref{eq:distribution of X-t-k+1 minus X-t-k} into~\eqref{eq:upper-bound I_5:monotonicity is used here} yields
{
\allowdisplaybreaks
\begin{align}
	\expectation{\rho}{\cI_5(\lambda)}
	&\leq \sup_{k\in\{0,1,\dots,K-1\}, t\in[0,T]} \expectation{\rho}{\exponential{\frac{\tilde{\lambda}}{\Delta_{\max}^2}\left\|X_{t} - X_{t_k}\right\|^2}} \\
	&= \sup_{k,t} \expectation{\mu,\xi^d}{ \exponential{\frac{\tilde{\lambda}}{\Delta_{\max}^2}
		\left\|\left(e^{-\bar{\Gamma}(t)} - e^{-\bar{\Gamma}(t_{k})}\right)X^\OU_0 + \bar{\sigma}_{t_k,t}\Xi^d\right\|^2}} \\
	&\leq \sup_{k,t} \expectation{\mu,\xi^d}{\exponential{
		2\frac{\tilde{\lambda}\left(e^{-\bar{\Gamma}(t)} - e^{-\bar{\Gamma}(t_{k})}\right)^2}{\Delta_{\max}^2}\left\|X^\OU_0\right\|^2 + 2\frac{\tilde{\lambda}\bar{\sigma}_{t_k,t}^2}{\Delta_{\max}^2}\left\|\Xi^d\right\|^2}} \\
	&\leq \sup_{k,t} \expectation{\mu,\xi^d}{\exponential{
		2\frac{\tilde{\lambda}\left(\Delta\tilde{\Gamma}_{t_k,t}\right)^2}{\Delta_{\max}^2}\left\|X^\OU_0\right\|^2
		+ 4\frac{\tilde{\lambda}\Delta\tilde{\Gamma}_{t_k,t}}{\Delta_{\max}^2}\left\|\Xi^d\right\|^2}} 
	\label{eq:upper-bound I_5:inequality of tilde-sigma-k is used here}\\
	&\leq  \expectation{\mu}{\exponential{2\tilde{\lambda}\left\|X^\OU_0\right\|^2}}
		\expectation{\xi^d}{ \exponential{4\frac{\tilde{\lambda}}{\Delta_{\max}}\left\|\Xi^d\right\|^2}} 
		\label{eq:upper-bound I_5:product of two expectations} \\ 
	&\leq \exponential{2\sigma^2d\tilde{\lambda} + \frac{2m^2\tilde{\lambda} 
		+ 4\sigma^4\tilde{\lambda}^2d}{1-4\sigma^2\tilde{\lambda}}}\cdot \left(1-8\frac{\tilde{\lambda}}{\Delta_{\max}}\right)^{-d/2} 
		\label{eq:upper-bound I_5:product of two expectations:upper bounds} \\
	&=
		\left(1 + 2\left(m^2 + \sigma^2d\right)\tilde{\lambda} + o\left(\left(m^2 + \sigma^2d\right)\tilde{\lambda}\right)\right)
		\left(1 + 4d\frac{\tilde{\lambda}}{\Delta_{\max}} + o\left(d\frac{\tilde{\lambda}}{\Delta_{\max}}\right)\right)
		\label{eq:upper-bound I_5:final result:-1}\\
	&=   1 + 4d\frac{\tilde{\lambda}}{\Delta_{\max}} +  2\left(m^2 + \sigma^2d\right)\tilde{\lambda} + o\left( d\frac{\tilde{\lambda}}{\Delta_{\max}} + \left(m^2+\sigma^2d\right)\tilde{\lambda} \right).
		\label{eq:upper-bound I_5:final result}
\end{align}
}Here, \eqref{eq:upper-bound I_5:inequality of tilde-sigma-k is used here} follows from~\eqref{eq:upper-bound of tilde-sigma-k} and the Lipschitz property of the function $x\mapsto e^{-x}$ for $x > 0$, while \eqref{eq:upper-bound I_5:product of two expectations} follows because $\Delta\tilde{\Gamma}_{t_k,t}\leq \Delta_{\max}$. To obtain~\eqref{eq:upper-bound I_5:final result:-1} we assume $d \tilde{\lambda} / \Delta_{\max} = o(1)$ and $\left(m^2 + \sigma^2d\right)\tilde{\lambda} = o(1)$. Note that the second condition has been used to establish~\eqref{eq:upper-bound I_4:upper bound of the first expectation}. 

\paragraph{Combine the Results}
Inserting~\eqref{eq:upper-bound I_3:final result}, \eqref{eq:upper-bound I_4:final result}, and \eqref{eq:upper-bound I_5:final result} into~\eqref{eq:I_1 is linked to the score error}, we obtain an upper bound of $\Phi(\lambda)$ via \eqref{eq:upper bound of the form prod I_i}.
{\allowdisplaybreaks
\begin{align}
	\Phi(\lambda)&\leq \frac{\tilde{\lambda}}{2\Delta_{\max}} + 
	\frac{3}{4}\log\expectation{\rho}{\exponential{4\lambda\tilde{\Gamma}(T)\scoreError(X)}}
	+ \frac{1}{8}\log  \left(1 + 8Ld\tilde{\lambda}+ o\left(Ld\tilde{\lambda}\right)\right) \nonumber \\
	\begin{split}
		&\qquad + \frac{1}{16}\log \left(
		1+ 32\left(m^2+(1+\sigma^2)d\right)\tilde{\lambda}+ o\left(\left(m^2+(1+\sigma^2)d\right)\tilde{\lambda}\right)
		\right) \\
		&\qquad + \frac{1}{16}\log  \left( 1 + 64d\frac{\tilde{\lambda}}{\Delta_{\max}} +  32\left(m^2 + \sigma^2d\right)\tilde{\lambda}+ o\left( d\frac{\tilde{\lambda}}{\Delta_{\max}} + \left(m^2+\sigma^2d\right)\tilde{\lambda}\right)\right)
	\end{split}\\
	&\lesssim \log\expectation{\rho}{\exponential{4\lambda\tilde{\Gamma}(T)\scoreError(X)}} + 
		(d+1)\frac{\tilde{\lambda}}{\Delta_{\max}} + (L+1)d\tilde{\lambda}+ \left(m^2+\sigma^2d\right)\tilde{\lambda}+ o\text{-term}
		\label{eq:logx is less than 1 plus x is used here} \\
	&\lesssim \log\expectation{\rho}{\exponential{4\lambda\tilde{\Gamma}(T)\scoreError(X)}} +
		d\frac{\tilde{\lambda}}{\Delta_{\max}} + Ld\tilde{\lambda}+ \left(m^2+\sigma^2d\right)\tilde{\lambda}+ o\text{-term}
		\label{eq:L is larger sim than 1 is used here} \\
	&\lesssim \log\expectation{\rho}{\exponential{4\lambda\tilde{\Gamma}(T)\scoreError(X)}} + \lambda L^2\widetilde{\Gamma}(T)
	\left( d  \Delta_{\max} + \left(m^2+\sigma^2d\right) \Delta_{\max}^2\right) + o\text{-term}.
	\label{eq:upper-bound the target in terms of lambda:final result} 
\end{align}
}In the above, \eqref{eq:logx is less than 1 plus x is used here} follows from $\log (1+x) \leq  x $ for $x \geq 0$, \eqref{eq:L is larger sim than 1 is used here} follows from $L\gtrsim 1$, and \eqref{eq:upper-bound the target in terms of lambda:final result} follows from $L\Delta_{\max}\lesssim 1$. This proves Proposition~\ref{proposition:upper-bound of Phi(lambda)}.

\section{Proof of Corollary~\ref{corollary:iteration complexity of DDPM}}
\label{subsection:proof of corollary:iteration complexity of DDPM}
\begin{proof}[Proof of Corollary~\ref{corollary:iteration complexity of DDPM}]
	We first control the mixing error $\cE^M$ by requiring $\cE^{M} =  \exponential{-\frac{2\tilde{\Gamma}(T)}{2\alpha-1}}\RenyiDivergence[2\alpha-1]{\mu}{\xi^d} \leq \epsilon$.
	This condition is satisfied whenever
	$\widetilde{\Gamma}(T) \geq \left(\alpha-\frac{1}{2}\right)\log\frac{ \RenyiDivergence[2\alpha-1]{\mu}{\xi^d}}{\epsilon}.$
	Accordingly, we choose 
	\begin{equation}
		\widetilde{\Gamma}(T) = C\alpha\log\left( 1 + \frac{\RenyiDivergence[2\alpha-1]{\mu}{\xi^d}}{\epsilon}\right),
	\end{equation}
    for some constant $C>1$, where the additional 1 inside the logarithm ensures $\widetilde{\Gamma}(T)>0$. 
	
	We next control the discretization error $\cE^D$. To simplify its upper
	bound, we first impose
	$(m^2+\sigma^2d) \Delta_{\max}^2\lesssim d \Delta_{\max} $, 
	so that the quadratic term in $\Delta_{\max}$ can be absorbed into the
	linear term. A sufficient condition is
	\begin{equation}
		\Delta_{\max} \lesssim \frac{d}{m^2} \wedge\frac{1}{\sigma^2}.
		\label{eq:constraint on Delta_max:1}
	\end{equation}
	Under this constraint, we further require $\frac{\alpha^2}{\alpha-1}L^2\widetilde{\Gamma}(T)d \Delta_{\max} \leq \epsilon$ which ensures that $\cE^D\leq \epsilon$. Substituting $\widetilde{\Gamma}(T)$ yields
	\begin{equation}
		 \Delta_{\max}\lesssim  \frac{(\alpha-1)\epsilon}{\alpha^3L^2d}\left(\log \left(1 + \dfrac{\RenyiDivergence[2\alpha-1]{\mu}{\xi^d}}{\epsilon}\right)\right)^{-1}. 
		 \label{eq:constraint on Delta_max:2}
	\end{equation}
	In addition, Lemma~\ref{lemma:score perturbation} requires
	$L\Delta_{\max}\lesssim 1$. Combining this condition with
	\eqref{eq:constraint on Delta_max:1} and
	\eqref{eq:constraint on Delta_max:2} establishes the provided $\Delta_{\max}$.
	Finally, by adopting a uniform discretization (\ie, $\Delta\widetilde{\Gamma}_k=\Delta_{\max}$ for every
	$k\in\{0,1,\ldots,K-1\}$), we derive the iteration complexity~\eqref{eq:iteration complexity of DDPM} via the relation $K =  \dfrac{\widetilde{\Gamma}(T)}{\Delta_{\max}}$.
%
\end{proof}

\section{Supplementary Materials for the Proof of Theorem~\ref{theorem:Renyi divergence error of Critically Damped DDPM}}
\label{section:appendix:proof of theorem:CD}

The key to completing the proof of Theorem~\ref{theorem:Renyi divergence error of Critically Damped DDPM} is to establish an analogous counterpart to Proposition~\ref{proposition:upper-bound of Phi(lambda)}, which we summarize as follows.

\begin{proposition}
	\label{proposition:upper-bound of Phi(lambda):critically damped}
	The upper bound derived for $\Phi(\lambda)$ in Proposition~\ref{proposition:upper-bound of Phi(lambda)} remains valid when replacing $\Phi(\lambda)$ with its critically damped counterpart, defined as
	\begin{equation}
		\label{eq:an expression that we want to upper-bound:critically damped}
		\Phi^\CD(\lambda)\coloneqq
		\log\expectation{(X,V)\sim\rho}{\exponential{\lambda \sum_{k=0}^{K-1} \int_{t_k}^{t_{k+1}}\bar{\gamma}(t)\left\|S^\CD_{t_k}\left(X_{t_k}, V_{t_k}\right) - \nabla_{V_t}\log \rho_{t}(X_t,V_t)\right\|^2 \d t}}.
	\end{equation}
\end{proposition}

The proof of this proposition follows the same technical trajectory as that of Proposition~\ref{proposition:upper-bound of Phi(lambda)}, albeit with more complex mathematical expressions. The remainder of this section focuses on the steps that differ, omitting details that carry over directly from the standard case.

According to Proposition~\ref{proposition:properties of CD}\ref{enumerate:properties of CDL:2},  we invoke the score perturbation lemma~\ref{lemma:score perturbation} with $\zeta = \cO\left(\Gamma(t)-\Gamma(s)\right)$ and $\|\Sigma\|_2 = \cO\left(\Gamma(t)-\Gamma(s)\right)$, yielding

\begin{equation}
	\Phi^\CD(\lambda)
	\leq \log\expectation{\rho}{\prod_{i=1}^{5}\cI^\CD_i(\lambda)}.
	\label{eq:upper bound of the form prod I_i:CD}
\end{equation}
Here, we define
{\allowdisplaybreaks
	\begin{align}
		\cI^\CD_1(\lambda) &\coloneqq \exponential{3\lambda\sum_{k=0}^{K-1}\int_{t_{k}}^{t_{k+1}}\bar{\gamma}(t)\left\|S^\CD_{t_k}\left(X_{t_k},V_{t_k}\right)-\nabla_{V_{t_k}}\log\rho_{t_k}\left(X_{t_k},V_{t_k}\right)\right\|^2\d t},\\
		\cI^\CD_2(\lambda) &\coloneqq \cI_2(\lambda) = 
		\exponential{\lambda CL^2\sum_{k=0}^{K-1}\int_{t_{k}}^{t_{k+1}}\bar{\gamma}(t) \Delta\tilde{\Gamma}_{t_k,t}\d t},\\
		\cI^\CD_3(\lambda) &\coloneqq
		\exponential{\lambda CL^2\sum_{k=0}^{K-1}\int_{t_{k}}^{t_{k+1}}\bar{\gamma}(t)\left\|\nabla_{V_t}\log\rho_{t}\left(X_{t},V_t\right)\right\|^2\left(\Delta\tilde{\Gamma}_{t_k,t}\right)^2\d t}\label{eq:definition of I_3:CD}\\
		\cI^\CD_4(\lambda) &\coloneqq
		\exponential{\lambda CL^2\sum_{k=0}^{K-1}\int_{t_{k}}^{t_{k+1}}\bar{\gamma}(t)\left(\left\|X_{t_k}\right\|^2 + \left\|V_{t_k}\right\|^2\right)\left(\Delta\tilde{\Gamma}_{t_k,t}\right)^2\d t},
		\label{eq:definition of I_4:CD}\\
		\cI^\CD_5(\lambda) &\coloneqq
		\exponential{\lambda CL^2\sum_{k=0}^{K-1}\int_{t_{k}}^{t_{k+1}}\bar{\gamma}(t)\left(\left\|X_t - X_{t_k}\right\|^2 + \left\|V_t - V_{t_k}\right\|^2\right)\d t}.
		\label{eq:definition of I_5:CD} 
	\end{align}
}
The first three terms can be handled analogously to their standard counterparts, the primary differences lie in $\cI^\CD_4$ and $\cI^\CD_5$.

\paragraph{Bound $\expectation{\rho}{\cI^\CD_4(\lambda)}$}

Define $A(t) = \begin{bmatrix} (1+\Gamma(t))e^{-\Gamma(t)} \\ -\Gamma(t)e^{-\Gamma(t)} \end{bmatrix}\otimes I_d$, where $\otimes$ denotes the Kronecker product.
The critically damped analogue of \eqref{eq:distribution of X-OU-t} is 
\begin{equation}
	\begin{bmatrix} X^\CD_t  \\ V^\CD_t  \end{bmatrix} \stackrel{\d}{=}
	A(t)  X^\CD_0 + \left(I_{2d} - A(t)A(t)^\top\right)^{1/2} \Xi^{2d},
\end{equation}
and hence 
\begin{equation}
	\begin{bmatrix} X_t  \\ V_t  \end{bmatrix} \stackrel{\d}{=} \begin{bmatrix} X^\CD_{T-t}  \\ V^\CD_{T-t}  \end{bmatrix} \stackrel{\d}{=}
	\bar{A}(t)  X^\CD_0 + \left(I_{2d} - \bar{A}(t)\bar{A}(t)^\top\right)^{1/2} \Xi^{2d}.
\end{equation}
Following a similar argument to that in~\eqref{eq:upper-bound I_4:move the sum inside the integral}-\eqref{eq:upper-bound I_4:monotonicity of Gamma-tilde is used here}, we obtain
\begin{align}
	\expectation{\rho}{\cI^\CD_4(\lambda)}
	&\leq \sup_{k\in\{0,1,\dots,K-1\}} \expectation{\rho}{\exponential{\tilde{\lambda}\left(\left\|X_{t_k}\right\|^2+\left\|V_{t_k}\right\|^2\right)}} \\
	&= \sup_{k\in\{0,1,\dots,K-1\}} \expectation{\mu,\xi^{2d}}{ \exponential{\tilde{\lambda}\left\|\bar{A}(t_k)X^\CD_0 + \left(I_{2d} - \bar{A}(t_k)\bar{A}(t_k)^\top\right)^{1/2} \Xi^{2d}\right\|^2}} \\
	&\leq \sup_{k\in\{0,1,\dots,K-1\}}  \expectation{\mu}{\exponential{2\tilde{\lambda} \left\|\bar{A}(t_k)X^\CD_0\right\|^2}}
	\expectation{\xi^{2d}}{ \exponential{2\tilde{\lambda}	\left\| \left(I_{2d} - \bar{A}(t_k)\bar{A}(t_k)^\top\right)^{1/2} \Xi^{2d}\right\|^2}}. 
	\label{eq:upper-bound I_4:product of two expectations:CD}
\end{align}
The first expectation admits the form
\begin{align}
	\expectation{\mu}{\exponential{2\tilde{\lambda} \left\|\bar{A}(t_k)X^\CD_0\right\|^2}}
	&= \expectation{\mu}{\exponential{2\tilde{\lambda}\left(1 + 2\bar{\Gamma}(t_k) + 2\bar{\Gamma}(t_k)^2\right)e^{-2\bar{\Gamma}(t_k)} \left\|X^\CD_0\right\|^2}} \\
	&\leq \expectation{\mu}{\exponential{2\tilde{\lambda} \left\|X^\CD_0\right\|^2}},
	\label{eq:upper-bound I_4:product of two expectations:term 1:CD}
\end{align}
where the inequality follows from the fact that $x\mapsto (1+2x+2x^2)e^{-2x} \le 1$ for all $x \geq 0$, with equality at $x=0$.
Notice that~\eqref{eq:upper-bound I_4:product of two expectations:term 1:CD} coincides with the first term in~\eqref{eq:upper-bound I_4:product of two expectations}, and hence it can be handled following an identical process.

To evaluate the second expectation in~\eqref{eq:upper-bound I_4:product of two expectations:CD}, we recall that for a centered Gaussian random variable $Z \sim \cN(0, \Sigma)$, its quadratic form possesses the moment generating function 
\begin{equation}
	\expectation{}{e^{t\|Z\|^2}} = \det (I_d - 2t\Sigma )^{-1/2} = \prod_{i=1}^d (1 - 2t\lambda_i)^{-1/2}, \quad \text{for } t \in \left(0, \frac{1}{2\lambda_{\max}(\Sigma)}\right),
	\label{eq:mgf of gaussian quadratic form}
\end{equation}
where $\lambda_i$ denotes the $i$-th eigenvalue of $\Sigma$. Applying~\eqref{eq:mgf of gaussian quadratic form} yields
{\allowdisplaybreaks
\begin{align}
	&\expectation{\xi^{2d}}{ \exponential{2\tilde{\lambda}	\left\| \left(I_{2d} - \bar{A}(t_k)\bar{A}(t_k)^\top\right)^{1/2} \Xi^{2d}\right\|^2}} = \det \left(I_{2d} - 4\tilde{\lambda} \left(I_{2d} - \bar{A}(t_k)\bar{A}(t_k)^\top\right) \right)^{-1/2} \\
	&= \det \left(\begin{bmatrix}
		1 - 4\tilde{\lambda} + 4\tilde{\lambda}\left(1+\bar{\Gamma}(t_k)\right)^2e^{-2\bar{\Gamma}(t_k)} & -4\tilde{\lambda}\bar{\Gamma}(t_k)\left(1+\bar{\Gamma}(t_k)\right)e^{-2\bar{\Gamma}(t_k)} \\
		-4\tilde{\lambda}\bar{\Gamma}(t_k)\left(1+\bar{\Gamma}(t_k)\right)e^{-2\bar{\Gamma}(t_k)} & 1 - 4\tilde{\lambda} + 4\tilde{\lambda} \bar{\Gamma}(t_k)^2e^{-2\bar{\Gamma}(t_k)} 
	\end{bmatrix}\otimes I_d\right)^{-1/2}\\
	&= \left(1-4\tilde{\lambda}\right)^{-d/2}\left(1-4\tilde{\lambda}\left(1 - \left(1 + 2\bar{\Gamma}(t_k) + 2\bar{\Gamma}(t_k)^2\right)e^{-2\bar{\Gamma}(t_k)} \right)\right)^{-d/2} \\
	&\leq \left(1-4\tilde{\lambda}\right)^{-d},
	\label{eq:upper-bound I_4:product of two expectations:term 2:CD}
\end{align}
}
where the inequality follows by the fact that $\left(1+2x+2x^2\right)e^{-2x}\geq 0$ for $x>0$. 
Notice that~\eqref{eq:upper-bound I_4:product of two expectations:term 2:CD} coincides with~\eqref{eq:upper-bound I_4:upper bound of the second expectation}, except for a constant factor difference in the dimension $d$ ($d$ versus $d/2$). This difference in the constant coefficient is ultimately absorbed by the $\lesssim$ notation.

\paragraph{Bound $\expectation{\rho}{\cI^\CD_5(\lambda)}$}
The critically damped counterpart to~\eqref{eq:distribution of X-OU-t minus X-OU-s} is
\begin{equation}
	\begin{bmatrix} X^\CD_t  \\ V^\CD_t  \end{bmatrix} - \begin{bmatrix} X^\CD_s  \\ V^\CD_s  \end{bmatrix} \stackrel{\d}{=}
	\left(A(t)-A(s)\right)  X^\CD_0 + C(s,t)^{1/2}\, \Xi^{2d},
\end{equation}
where $A(t) = \begin{bmatrix} (1+\Gamma(t))e^{-\Gamma(t)} \\ -\Gamma(t)e^{-\Gamma(t)} \end{bmatrix}\otimes I_d$ has been defined when bounding $\expectation{\rho}{\cI^\CD_4(\lambda)}$, and
$C(s,t) \coloneqq 2\left(I_{2d} - M(s,t)\right) - \left(A(t)-A(s)\right)\left(A(t)-A(s)\right)^\top$, with $M(s,t) \coloneqq \begin{bmatrix} 
	(1+\Delta\Gamma_{s,t})e^{-\Delta\Gamma_{s,t}} & 0 \\
	0 & (1-\Delta\Gamma_{s,t})e^{-\Delta\Gamma_{s,t}}\end{bmatrix} \otimes I_d$ and $\Delta\Gamma_{s,t} = \Gamma(t) - \Gamma(s)$.
Consequently, we have
\begin{equation}
	\begin{bmatrix} X_{t_{k}}  \\ V_{t_{k}}  \end{bmatrix} - \begin{bmatrix} X_{t}  \\ V_{t}  \end{bmatrix} 
	\stackrel{\d}{=}  \begin{bmatrix} X^\CD_{T-t_{k}}  \\ V^\CD_{T-t_{k}}  \end{bmatrix} - \begin{bmatrix} X^\CD_{T-t}  \\ V^\CD_{T-t}  \end{bmatrix}
	\stackrel{\d}{=} \left(\bar{A}(t_{k}) - \bar{A}(t)\right)  X^\CD_0 + \bar{C}(t, t_{k})^{1/2}\,\Xi^{2d}.
	\label{eq:distribution of X-t-k+1 minus X-t-k:CD}
\end{equation}
Recall in our notation, $\bar{C}(s,t)$ represents $C(T-s,T-t)$. Following a similar argument to that in~\eqref{eq:upper-bound I_5:move the sum inside the integral}-\eqref{eq:upper-bound I_5:monotonicity is used here}, we obtain
{
	\allowdisplaybreaks
	\begin{align}
		\expectation{\rho}{\cI^\CD_5(\lambda)}
		&\leq \sup_{k,t} \expectation{\rho}{\exponential{\frac{\tilde{\lambda}}{\Delta_{\max}^2}\left(\left\|X_{t} - X_{t_k}\right\|^2 + \left\|V_{t} - V_{t_k}\right\|^2\right)}} \\
		&= \sup_{k,t} \expectation{\mu,\xi^{2d}}{\exponential{\frac{\tilde{\lambda}}{\Delta_{\max}^2}\left\|\left(\bar{A}(t)-\bar{A}(t_{k})\right)  X^\CD_0 + \bar{C}(t,t_{k})^{1/2}\Xi^{2d}\right\|^2}} \\
		&\leq  \sup_{k,t} \expectation{\mu}{\exponential{\frac{2\tilde{\lambda}}{\Delta_{\max}^2}\left\|\left(\bar{A}(t)-\bar{A}(t_{k})\right)  X^\CD_0 \right\|^2}}
		\expectation{\xi^{2d}}{\exponential{\frac{2\tilde{\lambda}}{\Delta_{\max}^2}\left\|\bar{C}(t, t_{k})^{1/2}\Xi^{2d} \right\|^2}}.
		\label{eq:upper-bound I_5:product of two expectations:CD}
	\end{align}
}
Note that 
\begin{align}
	\begin{split}
		\left\|\left(\bar{A}(t)-\bar{A}(t_{k})\right)  X^\CD_0 \right\|^2 
		&= \left( (1+\bar{\Gamma}(t))e^{-\bar{\Gamma}(t)} - (1+\bar{\Gamma}(t_{k}))e^{-\bar{\Gamma}(t_{k})}\right)^2\left\|X^\CD_0 \right\|^2 
		\\
		&\quad\; + \left( \bar{\Gamma}(t)e^{-\bar{\Gamma}(t)} - \bar{\Gamma}(t_{k})e^{-\bar{\Gamma}(t_{k})}\right)^2 \left\|X^\CD_0 \right\|^2 
	\end{split} \\
	&\leq \left(\frac{1}{e^2}+1\right)\Delta_{\max}^2\left\|X^\CD_0 \right\|^2,
\end{align}
where the last step follows from the function $x\mapsto (1+x)e^{-x}$ is $\frac{1}{e}$-Lipschitz and the function $x\mapsto xe^{-x}$ is $1$-Lipschitz. 
The first expectation in~\eqref{eq:upper-bound I_5:product of two expectations:CD} admits the upper bound
\begin{equation}
	\expectation{\mu}{\exponential{\frac{2\tilde{\lambda}}{\Delta_{\max}^2}\left\|\left(\bar{A}(t)-\bar{A}(t_{k})\right)  X^\CD_0 \right\|^2}} \leq \expectation{\mu}{\exponential{\frac{2(1+e^2)\tilde{\lambda} }{e^2}\left\| X^\CD_0 \right\|^2}}. 
	\label{eq:upper-bound I_5:product of two expectations:term 1:CD}
\end{equation}
Note that the upper bound in~\eqref{eq:upper-bound I_5:product of two expectations:term 1:CD} differs from the first expectation in~\eqref{eq:upper-bound I_5:product of two expectations} only by a constant coefficient ($1 + 1/e^2$) of $\tilde{\lambda}$ , and this coefficient is ultimately absorbed in the $\lesssim$ notation.

Addressing the second expectation in~\eqref{eq:upper-bound I_5:product of two expectations:CD} is a little tricky, so we summarize the result in the following lemma, with the proof deferred to the end of this section.
\begin{lemma}
	Let $C(s,t)$ be defined as above and let $\lambda\in(0,1/4)$. Then
	\begin{equation}
		\det\left(I_{2d} - 2\lambda C(s,t)\right) \geq \left(1 - \left(\frac{4}{e}+8\right)\lambda\Delta\Gamma_{s,t}\right)^d.
	\end{equation}
	\label{lemma:lower bound of det(I-2lambda C)}
\end{lemma}
Applying Lemma~\ref{lemma:lower bound of det(I-2lambda C)} to the second expectation in~\eqref{eq:upper-bound I_5:product of two expectations:CD} yields
\begin{align}
	\expectation{\xi^{2d}}{\exponential{\frac{2\tilde{\lambda}}{\Delta_{\max}^2}\left\|\bar{C}(t, t_{k})^{1/2}\Xi^{2d} \right\|^2}}
	&= \det\left(I_{2d} - \frac{4\tilde{\lambda}}{\Delta_{\max}^2} C(T-t, T-t_k)\right)^{-1/2}\\
	&\leq \left(1 - 2\left(\frac{4}{e}+8\right)\frac{\tilde{\lambda}}{\Delta_{\max}}\right)^{-d/2}.
	\label{eq:upper-bound I_5:product of two expectations:term 2:CD} 
\end{align}
Again, the upper bound~\eqref{eq:upper-bound I_5:product of two expectations:term 2:CD} differs from its standard counterpart (the second term in~\eqref{eq:upper-bound I_5:product of two expectations:upper bounds}) only by a constant coefficient on $\tilde{\lambda}$ ($8$ versus $2(4 / e + 8)$), which is ultimately absorbed by the $\lesssim$ notation.

\begin{proof}[Proof of Lemma~\ref{lemma:lower bound of det(I-2lambda C)}]
	To simplify notations, we define 
	\begin{align}
		M^X &= 1-4\lambda\left(1-(1+\Delta\Gamma_{s,t})e^{-\Delta\Gamma_{s,t}}\right), \\
		M^V &= 1-4\lambda\left(1-(1-\Delta\Gamma_{s,t})e^{-\Delta\Gamma_{s,t}}\right), \\
		\psi^X &= a^X(t) - a^X(s) = (1+\Gamma(t))e^{-\Gamma(t)} - (1+\Gamma(s))e^{-\Gamma(s)}, \\
		\psi^V &= a^V(t) - a^V(s) = \Gamma(t) e^{-\Gamma(t)} - \Gamma(s)e^{-\Gamma(s)}.
	\end{align}
	Then we are ready to calculate the determinant
	{\allowdisplaybreaks
	\begin{align}
		\det(I_{2d}-2\lambda C(s,t))
		&= \det \Big(\underbrace{(1-4\lambda)I_{2d} + 4\lambda M(s,t)}_{\text{denote by } \widetilde{M}} + 2\lambda \underbrace{\left(A(t)-A(s)\right)}_{\text{denote by }\psi}\left(A(t)-A(s)\right)^\top \Big) \\
		&= \det(\widetilde{M}) \det\left(I_d+2\lambda \psi^\top \widetilde{M}^{-1} \psi \right) \\
		&= \Big(M^X M^V + \underbrace{2\lambda \left(\left(\psi^X\right)^2M^V + \left(\psi^V\right)^2M^X\right)}_{\geq 0 \text{ for }\lambda\in(0,1/4)}\Big)^d \\
		&\geq \left(M^X M^V\right)^d \\
		&= \big(1-4\lambda\underbrace{\left(1-(1+\Delta\Gamma_{s,t})e^{-\Delta\Gamma_{s,t}}\right)}_{\leq \frac{1}{e}\Delta\Gamma_{s,t}}\big)^d\big(1-\underbrace{4\lambda\left(1-(1-\Delta\Gamma_{s,t})e^{-\Delta\Gamma_{s,t}}\right)}_{\leq 2\Delta\Gamma_{s,t}}\big)^d \\
		&\geq \left(1 - \frac{4}{e}\lambda \Delta \Gamma_{s,t}\right)^d\left(1 - 8 \lambda \Delta \Gamma_{s,t}\right)^d 
		\label{eq:proof of lemma:lower bound of det(I-2lambda C):1}\\
		&\geq \left(1 - \left(\frac{4}{e}+8\right)\lambda \Delta \Gamma_{s,t}\right)^d.
	\end{align}
	}Here, \eqref{eq:proof of lemma:lower bound of det(I-2lambda C):1} follows because the Lipschitz constants of functions $x\mapsto (1+x)e^{-x}$ and $x\mapsto (1-x)e^{-x}$ ($x > 0$) are $\frac{1}{e}$ and $2$, respectively.
\end{proof}

\section{Background on Linear \ac{sde} and \ac{ou} Processes}
\label{section:appendix:backgrounds on linear SDE}
This section provides background on time-rescaled linear SDEs, followed by essential properties of standard and kinetic \ac{ou} processes utilized throughout our analysis.

\subsection{Solution to Linear \ac{sde}}
\label{subsection:solution to linear sde}
Consider a time-scaled linear \ac{sde}
\begin{equation}
	\d Y_t = \gamma(t)AY_t \d t+ \sqrt{\gamma(t)}B\d W_t.
	\label{eq:time-scaled linear sde}
\end{equation}
The associated integrating factor is $e^{-\Gamma(t) A}$, where $\Gamma(t) = \int_0^t \gamma(\tau) \d \tau$. Defining $Z_t = e^{-\Gamma(t)A}Y_t$, an application of Itô's formula yields
\begin{align}
	\d Z_t &= -\gamma(t) A e^{-\Gamma(t) A} Y_t \d t + e^{-\Gamma(t) A} \left( \gamma(t) A Y_t \d t + \sqrt{\gamma(t)} B \d W_t \right)
	\label{eq:appendix:solution to linear SDE:intermediate result 1} \\
	& = e^{-\Gamma(t) A} \sqrt{\gamma(t)} B \d W_t,
\end{align}
where the drift terms in~\eqref{eq:appendix:solution to linear SDE:intermediate result 1} cancel out because $A$ commutes with $e^{-\Gamma(t) A}$. Integrating over the time interval $[s,t]$ gives
\begin{equation}
	Z_t = Z_s + \int_s^t e^{-\Gamma(\tau) A} \sqrt{\gamma(\tau)} B \d W_\tau,
\end{equation}
and hence
\begin{equation}
	Y_t = e^{\left(\Gamma(t)-\Gamma(s)\right) A} Y_s + \int_s^t e^{(\Gamma(t) - \Gamma(\tau)) A} \sqrt{\gamma(\tau)} B \d W_\tau.
	\label{eq:appendix:solution to Y_t}
\end{equation}
In particular, setting $s = 0$ reduces~\eqref{eq:appendix:solution to Y_t} to $\displaystyle Y_t = e^{\Gamma(t) A} Y_0 + \int_0^t e^{(\Gamma(t) - \Gamma(\tau)) A} \sqrt{\gamma(\tau)} B \d W_\tau.$

\subsection{\ac{ou} process (Over Damped Langevin Dynamics)}
For the \ac{ou} process~\eqref{eq:OU process}, this corresponds to setting $A = -I_d$ and $B = \sqrt{2}I_d$. Below, we collect several key properties of the \ac{ou} process utilized throughout our analysis.

\begin{proposition}[Properties of the \ac{ou} process]
	\label{proposition:properties of OU}
	Let $X^\OU_t$ be the \ac{ou} process~\eqref{eq:OU process}. For $0<s<t$,
	\begin{enumerate}[label=(\alph*)]
		\item The solution to $X^\OU_t$ is given by $\displaystyle X^\OU_t = e^{-\left(\Gamma(t)-\Gamma(s)\right) } X^\OU_s + e^{-\Gamma(t)}\int_s^t \sqrt{2\gamma(\tau)} e^{\Gamma(\tau)}  \d W_\tau$. In particular, 
		$\displaystyle X^\OU_t = e^{-\Gamma(t) } X^\OU_0 + e^{-\Gamma(t)}\int_0^t \sqrt{2\gamma(\tau)} e^{\Gamma(\tau)}  \d W_\tau$.
		\label{enumerate:properties of OU process:1}
		\item $X^\OU_t\stackrel{\d}{=}e^{ -\left(\Gamma(t) - \Gamma(s)\right)} X^\OU_s + \sqrt{1-e^{-2\left(\Gamma(t) - \Gamma(s)\right)}}\,\Xi^d$, where $\Xi^d\sim\xi^d$. In particular, $X^\OU_t\stackrel{\d}{=}e^{ -\Gamma(t) } X^\OU_0 + \sqrt{1-e^{-2\Gamma(t)}}\,\Xi^d$.
		\label{enumerate:properties of OU process:2}
		\item $X^\OU_t - X^\OU_s \stackrel{\d}{=} \left(e^{-\Gamma(t)} - e^{-\Gamma(s)}\right)X^\OU_0 + \sqrt{2\left(1-e^{-(\Gamma(t)-\Gamma(s))}\right) - \left( e^{-\Gamma(t)} - e^{-\Gamma(s)}\right)^2}\,\Xi^d$.
		\label{enumerate:properties of OU process:3}
	\end{enumerate}
\end{proposition}
In the above, \ref{enumerate:properties of OU process:1} is a consequence of~\eqref{eq:appendix:solution to Y_t}. Both \ref{enumerate:properties of OU process:2} and \ref{enumerate:properties of OU process:3} follow from It\^o's isometry. Here, we prove \ref{enumerate:properties of OU process:3}, and \ref{enumerate:properties of OU process:2} follows by an analogous argument.
\begin{proof}
	Using the solution to OU process, we have
	$X^\text{OU}_t - X^\text{OU}_s = \left(e^{-\Gamma(t)} - e^{-\Gamma(s)}\right)X^\text{OU}_0 + I_{s,t}$, where
	\begin{align}
		I_{s,t} &= \int_{0}^{t} \sqrt{2\gamma_\tau}e^{-(\Gamma(t)-\Gamma(\tau))} \d W_\tau - \int_{0}^{s} \sqrt{2\gamma_\tau}e^{-(\Gamma(s)-\Gamma(\tau))} \d W_\tau \\
		&= \underbrace{\left(e^{-\Gamma(t)} - e^{-\Gamma(s)}\right)\int_{0}^{s} \sqrt{2\gamma_\tau}e^{\Gamma(\tau)} \d W_\tau}_{\text{Term 1}} + \underbrace{\int_{s}^{t} \sqrt{2\gamma_\tau}e^{-(\Gamma(t)-\Gamma(\tau))} \d W_\tau}_{\text{Term 2}}.
	\end{align}
	Since $I_{s,t}$ is a centered isotropic Gaussian vector, we write by $I_{s,t}\stackrel{\d}{=}\sigma_{s,t}\,\Xi^d$, and it remains to determine its variance $\sigma^2_{s,t}$. The It\^o's isometry implies
	{
	\allowdisplaybreaks
	\begin{align}
		\text{Var}(\text{Term 1}) 
		&= \left(e^{-\Gamma(t)} - e^{-\Gamma(s)}\right)^2 \int_{0}^{s} 2\gamma_\tau e^{2\Gamma(\tau)} \d\tau 
		= \left(e^{-\Gamma(t)} - e^{-\Gamma(s)}\right)^2 \left[e^{2\Gamma(\tau)}\right]_0^s \\
		&= \left(e^{-\Gamma(t)} - e^{-\Gamma(s)}\right)^2 \left(e^{2\Gamma(s)} - 1\right), \\
		\text{Var}(\text{Term 2}) 
		&= 2e^{-2\Gamma(t)}\int_{s}^{t} \gamma_\tau e^{2\Gamma(\tau)} \d\tau 
		= 2e^{-2\Gamma(t)} \left[\frac{1}{2}e^{2\Gamma(\tau)}\right]_s^t 
		= 1 - e^{-2(\Gamma(t)-\Gamma(s))}.
	\end{align}
	}  
	Adding the two terms together yields
	\begin{align}
		\sigma_{s,t}^2 &= \left(e^{-\Gamma(t)} - e^{-\Gamma(s)}\right)^2 \left(e^{2\Gamma(s)} - 1\right) + 1 - e^{-2(\Gamma(t)-\Gamma(s))} \\
		&= 2\left(1 - e^{-(\Gamma(t)-\Gamma(s))}\right) - \left(e^{-\Gamma(t)} - e^{-\Gamma(s)}\right)^2.
	\end{align}
	This proves~\ref{enumerate:properties of OU process:3}.
\end{proof}
	
\subsection{Kinetic \ac{ou} Process (Critically Damped Langevin Dynamics)}
The kinetic \ac{ou} process is characterized by \acp{sde}
\begin{equation}
	\left\{
	\begin{aligned}
	\d X^\CD_t &= \gamma(t)V^\CD_t\d t, \\
	\d V^\CD_t &= -\gamma(t)\left(X^\CD_t + 2V^\CD_t\right)\d t + 2\sqrt{\gamma(t)}\d W_t.
	\end{aligned}
	\right.
\end{equation}
Several useful properties of kinetic \ac{ou} process is collected in the following proposition.
\begin{proposition}[Properties of the Kinetic \ac{ou} Process]
	\label{proposition:properties of CD}
	Let $[X^\CD_t, V^\CD_t]^\top$ be the kinetic \ac{ou} process. For $0<s<t$,
	\begin{enumerate}[label=(\alph*)]
		\item The solution to $[X^\CD_t, V^\CD_t]^\top$ is given by 
		\begin{equation}
			\begin{bmatrix} X^\CD_t  \\ V^\CD_t  \end{bmatrix}
			= A(s,t) \begin{bmatrix} X^\CD_s  \\ V^\CD_s  \end{bmatrix} 
			+ \int_s^t  \sqrt{\gamma(\tau)} e^{-\Delta\Gamma_{\tau,t}} \begin{bmatrix} 0 & 2\Delta\Gamma_{\tau,t} I \\ 0 & 2(1-\Delta\Gamma_{\tau,t})I \end{bmatrix} \d W_\tau,
		\end{equation}
		where $\Delta\Gamma_{\tau,t} \coloneqq \Gamma(t) - \Gamma(\tau)$ and $\displaystyle A(s,t) \coloneqq e^{-\Delta\Gamma_{s,t}} \begin{bmatrix} (1+\Delta\Gamma_{s,t}) & \Delta\Gamma_{s,t}  \\ -\Delta\Gamma_{s,t}  & (1-\Delta\Gamma_{s,t}) \end{bmatrix}\otimes I_d$.
		\label{enumerate:properties of CDL:1}
		\item $\displaystyle \begin{bmatrix} X^\CD_t  \\ V^\CD_t  \end{bmatrix} \stackrel{\d}{=}  A(s,t) \begin{bmatrix} X^\CD_s  \\ V^\CD_s  \end{bmatrix}
		+ \Sigma(s,t)^{1/2}\xi^{2d}$, where $\Sigma(s,t) =  \begin{bmatrix} \Sigma_{XX}(s,t)  & \Sigma_{XV}(s,t)  \\ \Sigma_{VX}(s,t)  & \Sigma_{VV}(s,t)  \end{bmatrix}\otimes I_d$, with
		\begin{align}
			\Sigma_{XX}(s,t) &= 1 - e^{-2\Delta\Gamma_{s,t}}\left(2\left(\Delta\Gamma_{s,t}\right)^2 + 2\Delta\Gamma_{s,t} + 1\right), \\
			\Sigma_{XV}(s,t) &= \Sigma_{VX}(s,t) = 2\left(\Delta\Gamma_{s,t}\right)^2 e^{-2\Delta\Gamma_{s,t}}, \\
			\Sigma_{VV}(s,t) &= 1 - e^{-2\Delta\Gamma_{s,t}}\left(2\left(\Delta\Gamma_{s,t}\right)^2 - 2\Delta\Gamma_{s,t} + 1\right).
		\end{align}
		In particular, $A(s,t)$ has only one eigenvalue $\lambda_A = e^{-\Delta\Gamma_{s,t}}$, and the spectrum of $\Sigma(s,t)$ are given by 
		$\displaystyle 	\lambda_{\Sigma}^{\pm} = 1 - e^{-2\Delta\Gamma_{s,t}} \left(\sqrt{\left(\Delta\Gamma_{s,t}\right)^2 + 1} \mp \Delta\Gamma_{s,t}\right)^2$.
		\label{enumerate:properties of CDL:2}
	\end{enumerate}
	In what follows, we assume $V^\CD_0\sim \xi^d$.
	\begin{enumerate}[resume, label=(\alph*)]
		\item $X^\CD_t \stackrel{\d}{=} a^X(t)X^\CD_0 + \sqrt{1 - a^X(t)^2}\,\Xi^d$, where $a^X(t) = e^{-\Gamma(t)}(1+\Gamma(t))$.
		\label{enumerate:properties of CDL:3}
		
		\item $V^\CD_t \stackrel{\d}{=} -a^V(t)X^\CD_0 + \sqrt{1 - a^V(t)^2}\,\Xi^d$, where $a^V(t) = e^{-\Gamma(t)}\Gamma(t)$.
		\label{enumerate:properties of CDL:4}
		
		\item Let $a^X(\cdot)$ and $a^V(\cdot)$ be defined as in \ref{enumerate:properties of CDL:3} and \ref{enumerate:properties of CDL:4}, respectively. Denote by
		$A(t) = \left[a^X(t),-a^V(t)\right]^\top \otimes I_d$, then 
		\begin{equation}
			\begin{bmatrix} X^\CD_t  \\ V^\CD_t  \end{bmatrix} \stackrel{\d}{=}
			A(t)  X^\CD_0 + \left(I_{2d} - A(t)A(t)^\top\right)^{1/2}\,\Xi^{2d}.
		\end{equation}
		\label{enumerate:properties of CDL:5}
		
		\item Define $a^X(\cdot)$ and $a^V(\cdot)$ as in \ref{enumerate:properties of CDL:3} and \ref{enumerate:properties of CDL:4}, respectively.
		\begin{equation}
			\begin{bmatrix} X^\CD_t  \\ V^\CD_t  \end{bmatrix} - \begin{bmatrix} X^\CD_s  \\ V^\CD_s  \end{bmatrix} \stackrel{\d}{=}
			\left(A(t)-A(s)\right) X^\CD_0 + C(s,t)^{1/2}\,\Xi^{2d},
		\end{equation}
		where $C(s,t) = 2\left(I_{2d} - M(s,t)\right) - \left(A(t)-A(s)\right)\left(A(t)-A(s)\right)^\top$, with $M(s,t)$ defined by 
		\begin{equation*}
			M(s,t) = \begin{bmatrix} 
				(1+\Delta\Gamma_{s,t})e^{-\Delta\Gamma_{s,t}} & 0 \\
				0 & (1-\Delta\Gamma_{s,t})e^{-\Delta\Gamma_{s,t}}\end{bmatrix} \otimes I_d.
		\end{equation*}
		\label{enumerate:properties of CDL:6}
	\end{enumerate}
\end{proposition}

\begin{proof}
	\ref{enumerate:properties of CDL:1} is a consequence of~\eqref{eq:appendix:solution to Y_t}, with
	\begin{equation}
		Y_t = 
		\begin{bmatrix}
			X_t \\
			V_t
		\end{bmatrix},
		A = 
		\begin{bmatrix}
			\bzero & I_d\\
			-I_d & -2I_d
		\end{bmatrix},
		B = 
		\begin{bmatrix}
			\bzero & \bzero\\
			\bzero & 2I_d
		\end{bmatrix}.
		\label{eq:appendix:expression of Y_t, A, and B}
	\end{equation}
	To calculate the exponential of the matrix $A$, we write by $A = -I_{2d} + N$, where $N = \begin{bmatrix} I_d & I_d \\ -I_d & -I_d \end{bmatrix}$. Since $N^2 = 0$, we have, for any $\lambda\in\bbR$,
	\begin{equation}
		e^{\lambda A} = e^{\lambda \left(-I_{2d} + N\right)} = e^{-\lambda I_{2d}}e^{\lambda N} = e^{-\lambda}\left(I_{2d} + \lambda N\right) 
		= e^{-\lambda} \begin{bmatrix} (1+\lambda)I_d &\lambda I_d \\ -\lambda I_d & (1-\lambda)I_d \end{bmatrix}.
		\label{eq:appendix:solution to e^Gamma(t)A}
	\end{equation}
	Inserting~\eqref{eq:appendix:expression of Y_t, A, and B} and \eqref{eq:appendix:solution to e^Gamma(t)A} into~\eqref{eq:appendix:solution to Y_t} yields \ref{enumerate:properties of CDL:1}. 
	
	The main ingredient in proving items~\ref{enumerate:properties of CDL:2}--\ref{enumerate:properties of CDL:6} is It\^o's isometry. The proof follows an argument analogous to that of Proposition~\ref{proposition:properties of OU}\ref{enumerate:properties of OU process:3}.
\end{proof}

\section{Mathematical Derivation in Fig.~\ref{fig:Renyi divergence is more sensitive} and Table~\ref{tab:metrics_scaling}}
\label{section:appendix:details_of_the_figure}
This section provides explicit mathematical formulas for the discrepancy measures plotted in Fig.~\ref{fig:Renyi divergence is more sensitive} and the scales summarized in Table~\ref{tab:metrics_scaling}. 

Recall that $P = (1-\epsilon)\cN(0,1) + \epsilon\cN(m,1)$ and $Q = \cN(0,1)$. Without loss of generality, we assume $m > 0$. Let $p(x)$ and $q(x) = \frac{1}{\sqrt{2\pi}}e^{-x^2/2}$ denote the probability density functions of $P$ and $Q$, respectively. We express $p(x)$ in terms of its main component $p_1(x) \coloneqq q(x)$ and tail component $p_2(x) \coloneqq \frac{1}{\sqrt{2\pi}}e^{-(x-m)^2/2}$ as
\begin{equation}
	p(x) = (1-\epsilon)p_1(x) + \epsilon p_2(x) = \frac{1-\epsilon}{\sqrt{2\pi}} e^{-x^2/2} + \frac{\epsilon}{\sqrt{2\pi}} e^{-(x-m)^2/2}.
\end{equation}
Additionally, the cumulative distribution functions of $P$ and $Q$ are denoted by $F_P$ and $F_Q$, respectively.

\paragraph{\ac{tv} Distance}
\begin{equation}
	\totalVariation{P}{Q} = \frac{1}{2}\int_{-\infty}^{\infty} \vert p(x) - q(x)\vert \d x =  \frac{\epsilon}{2} \int_{-\infty}^{\infty} |p_2(x) - q(x)| \d x 
	= \epsilon\left(F_Q\left(\frac{m}{2}\right) - F_Q\left(-\frac{m}{2}\right)\right)
	= \epsilon\left(2F_Q\left(\frac{m}{2}\right)-1\right).
\end{equation}

\paragraph{Wasserstein Distance} The Wasserstein distances for 1-dimensional probability distributions admit closed-form expressions~\cite{villani2009optimal}.

\begin{equation}
	\WassersteinDistance[1]{P}{Q} = \int_{-\infty}^{\infty} \vert F_P(x) - F_Q(x)\vert \d x = \int_{-\infty}^{\infty} \left(F_Q(x) - F_P(x)\right) \d x
	= \expectation{X\sim P}{X} - \expectation{Y\sim Q}{Y} = \epsilon m,
\end{equation}
\begin{equation}
	\squaredWassersteinDistance[2]{P}{Q} =  \int_0^1 \left\vert F_P^{-1}(x) - F_Q^{-1}(x) \right\vert^2 \d x \stackrel{x=F_P(y)}{=}
	\int_{-\infty}^{\infty} \left\vert y - F_Q^{-1}(F_P(y)) \right\vert^2 p(y) \d y. 
	\label{eq:Wasserstein distance:2}
\end{equation}
We use \eqref{eq:Wasserstein distance:2} for numerical computation. In particular, by translating the tail component of $P$ by $m$, we obtain the upper bound $\squaredWassersteinDistance{P}{Q} \leq \epsilon m^2$.

\paragraph{KL Divergence}
\begin{equation}
	\KLDivergence{P}{Q} = \underbrace{\int_{-\infty}^{\infty} p(x) \log p(x) \d x}_{-\text{Entropy of } P, \ -H(P)} - \underbrace{\int_{-\infty}^{\infty} p(x) \log q(x) \d x}_{\text{Cross Entropy}, \ H(P,Q)}. 
\end{equation}
The cross entropy admits the closed-form expression $H(P,Q) = \frac{1}{2}\ln(2\pi) + \frac{1}{2}(1 + \epsilon m^2)$, while the entropy term $H(P)$ is evaluated numerically.
Furthermore, by Jensen's inequality we obtain the upper-bound $\KLDivergence{P}{Q}\leq \epsilon \KLDivergence{\cN(m,1)}{\cN(0,1)} = \frac{1}{2}\epsilon m^2$.

\paragraph{\Renyi Divergence} For any integer order $\alpha \geq 2$, we have
\begin{align}
	\RenyiDivergence{P}{Q} &=  \frac{1}{\alpha - 1} \log \int_{-\infty}^{\infty} q(x) \left( \frac{p(x)}{q(x)} \right)^\alpha \d x\\
	&= \frac{1}{\alpha - 1} \log \int_{-\infty}^{\infty} q(x) \left( (1 - \epsilon) + \epsilon e^{m x - \frac{1}{2}m^2} \right)^\alpha \d x \\
	&= \frac{1}{\alpha - 1} \log \left( \sum_{k=0}^\alpha \binom{\alpha}{k} (1 - \epsilon)^{\alpha-k} \epsilon^k \int_{-\infty}^{\infty} \frac{1}{\sqrt{2\pi}} 
	e^{-\frac{1}{2}x^2 + km x - \frac{1}{2}km^2}\d x \right) \\
	&= \frac{1}{\alpha - 1} \log \left( \sum_{k=0}^\alpha \binom{\alpha}{k} (1 - \epsilon)^{\alpha-k} \epsilon^k \exp\left( \frac{k(k - 1)}{2} m^2 \right) \right).
\end{align}
For large $m$, the $k = \alpha$ term dominates the summation, yielding the asymptotic expression
\begin{align}
	\RenyiDivergence{P}{Q} &= \frac{1}{\alpha - 1} \log \left( \epsilon^\alpha \exp\left( \frac{\alpha(\alpha - 1)}{2} m^2 \right) \left(1 +  \cO\left(e^{-(\alpha-1)m^2}\right)\right)\right) \\
	&= \frac{\alpha}{2} m^2 - \frac{\alpha}{\alpha - 1} \log\left(\frac{1}{\epsilon}\right) + \cO\left(e^{-(\alpha-1)m^2}\right).
\end{align}

%
%
%
%